\documentclass[11pt]{article}%
\usepackage{amsfonts}
\usepackage{amsmath}
\usepackage{fullpage}
\usepackage{amssymb}
\usepackage{graphicx}
\usepackage[colorlinks=true,linkcolor=blue,citecolor=red,plainpages=false,pdfpagelabels]%
{hyperref}%
\providecommand{\U}[1]{\protect\rule{.1in}{.1in}}
\newtheorem{theorem}{Theorem}

\newtheorem{definition}[theorem]{Definition}

\newtheorem{lemma}[theorem]{Lemma}

\newtheorem{proposition}[theorem]{Proposition}
\newtheorem{remark}[theorem]{Remark}

\newenvironment{proof}[1][Proof]{\noindent\textbf{#1.} }{\ \rule{0.5em}{0.5em}}
\numberwithin{equation}{section}
\begin{document}

\title{\textbf{Amortized entanglement of a quantum channel and approximately
teleportation-simulable channels}}
\author{Eneet Kaur\thanks{Department of Physics and Astronomy, Louisiana State
University, Baton Rouge, Louisiana 70803, USA}
\and Mark M. Wilde\footnotemark[1] \thanks{Center for Computation and Technology,
Louisiana State University, Baton Rouge, Louisiana 70803, USA}}
\maketitle

\begin{abstract}
This paper defines the amortized entanglement of a quantum channel as the
largest difference in entanglement between the output and the input of the
channel, where entanglement is quantified by an arbitrary entanglement
measure. We prove that the amortized entanglement of a channel obeys several
desirable properties, and we also consider special cases such as the amortized
relative entropy of entanglement and the amortized Rains relative entropy. These latter quantities are shown to be single-letter upper bounds on the
secret-key-agreement and PPT-assisted quantum capacities of
a quantum channel, respectively.
Of
especial interest is a uniform continuity bound for these latter two special
cases of amortized entanglement, in which the deviation between the amortized
entanglement of two channels is bounded from above by a simple function of the
diamond norm of their difference and the output dimension of the channels. We
then define approximately teleportation- and
positive-partial-transpose-simulable (PPT-simulable) channels as those that
are close in diamond norm to a channel which is either exactly teleportation-
or PPT-simulable, respectively. These results then lead to single-letter upper
bounds on the secret-key-agreement and PPT-assisted quantum capacities of
channels that are approximately teleportation- or PPT-simulable, respectively.
Finally, we generalize many of the concepts in the paper to the setting of
general resource theories, defining the amortized resourcefulness of a channel
and the notion of $\nu$-freely-simulable channels, connecting these concepts
in an operational way as well.

\end{abstract}

\section{Introduction}

Evaluating or determining bounds on the various communication capacities of a
quantum channel is one of the main concerns of quantum information theory
\cite{H13book,W15book}. One can consider supplementing a channel with an
additional resource such as free entanglement
\cite{PhysRevLett.83.3081,ieee2002bennett}\ or classical communication
\cite{BBPSSW96EPP,BDSW96}, and such a consideration leads to different kinds
of capacities. Supplementing a channel with free classical communication, with
the goal being to communicate quantum information or private classical
information reliably, is of particular relevance due to its connection with
the operational setting of quantum key distribution \cite{bb84,E91}. The
former is called the local operations and classical communication (LOCC)
assisted quantum capacity, while the latter is called the secret-key-agreement capacity.

The relevance of these latter capacities is that an upper bound on them can
serve as a benchmark to determine whether one has experimentally implemented a
working quantum repeater \cite{L15}, which is a device needed for the
practical implementation of quantum key distribution. A first result in this
direction, building on earlier developments in \cite{CW04,C06,HHHO05,HHHO09},
is due to \cite{TGW14IEEE,TGW14Nat} (see also \cite{Wilde2016}), in which it
was shown that the squashed entanglement of a quantum channel is an upper
bound on both its LOCC-assisted quantum capacity and its secret-key-agreement
capacity. Some follow-up works \cite{PLOB15,WTB16}\ then considered other
entanglement measures such as relative entropy of entanglement and established
their relevance as bounds on these capacities in certain cases. There has been
an increasing interest in this topic in recent years, with a series of papers
developing it further
\cite{TGW14IEEE,TGW14Nat,STW16,PLOB15,Goodenough2015,TSW16,AML16,WTB16,Christandl2017,Wilde2016,BA17,RGRKVRHWE17,KW17,TSW17,RKBKMA17}%
.

In this paper, we develop this topic even further, in the following ways:

\begin{enumerate}
\item First, we define the amortized entanglement of a quantum channel as the
largest difference in the entanglement between the output and the input of the
channel, with entanglement quantified by some entanglement measure
\cite{H42007}. We note that amortized entanglement is closely related to ideas
put forth in \cite{BHLS03,LHL03,Christandl2017,BGMW17}, which were used
therein to give bounds on the performance of adaptive protocols (see also the
very recent paper \cite{RKBKMA17}\ for related ideas).

\item We then prove several properties of the amortized entanglement, while
considering special cases in which the entanglement measure is set to the
relative entropy of entanglement \cite{VP98}\ or the Rains relative entropy
\cite{R99,R01,AdMVW02}.
These latter quantities are shown to be single-letter upper bounds on the
secret-key-agreement and PPT-assisted quantum capacities of
a quantum channel, respectively.
Another important property that we establish in these
special cases is that the amortized entanglement obeys a uniform continuity
bound of the flavor in \cite{Winter15,S16cont}, with a dependence on the
output dimension of the two channels under consideration and the diamond norm
of their difference \cite{K97}.

\item These latter results lead to upper bounds on the secret-key-agreement
capacity of approximately teleportation-simulable channels (channels that are
close in diamond norm to a teleportation-simulable channel \cite{BDSW96,HHH99,CDP09}%
). Similarly, we find upper bounds on the positive-partial-transpose
(PPT)\ assisted quantum capacity for approximately PPT-simulable channels
(defined later). The main idea behind obtaining these bounds is broadly
similar to the approach of approximately degradable channels put forth in
\cite{SSWR14}.

\item We next showcase the aforementioned bounds for a simple qubit channel,
which is a convex combination of an amplitude damping channel and a
depolarizing channel (note that this channel is considered in the concurrent
work \cite{LKDW17} as well). The main finding here is that the upper bounds
from approximate simulation are reasonably close to lower bounds on the
capacities whenever the noise in the channel is low, and this result is
consistent with that which was found in earlier work \cite{SSWR14,LLS17}.

\item Finally, we discuss how many of the concepts developed in our paper can
be extended to general resource theories \cite{BG15,fritz_2015,RKR15,KR16}. In
particular, we discuss the amortized resourcefulness of a quantum channel and
prove how it leads to an upper bound on the amount of resourcefulness that can
be extracted from multiple calls to a quantum channel by interleaving calls to
it with free channels. We also introduce the notion of a $\nu$%
-freely-simulable channel as a generalization of the concept of a
teleportation-simulable channel.
\end{enumerate}

At the end of the paper, we conclude with a summary and open questions. The
rest of our paper proceeds in the order given above,
Appendix~\ref{sec:supp-lemmas} provides some supplementary lemmas that are
needed to establish the uniform continuity bound mentioned above, and
Appendix~\ref{sec:cov-TP-sim} discusses the relation between approximate
covariance \cite{LKDW17} and approximately teleportation-simulable channels, as well as
showing how to simulate the twirling of a channel \cite{BDSW96} via a generalized
teleportation protocol. Throughout our paper, we use notation and concepts
that are by now standard in quantum information theory, and we point the
reader to \cite{W15book}\ for background.

\section{Amortized entanglement of a quantum channel}

\begin{figure}[ptb]
\begin{center}
\includegraphics[
width=3.5405in
]{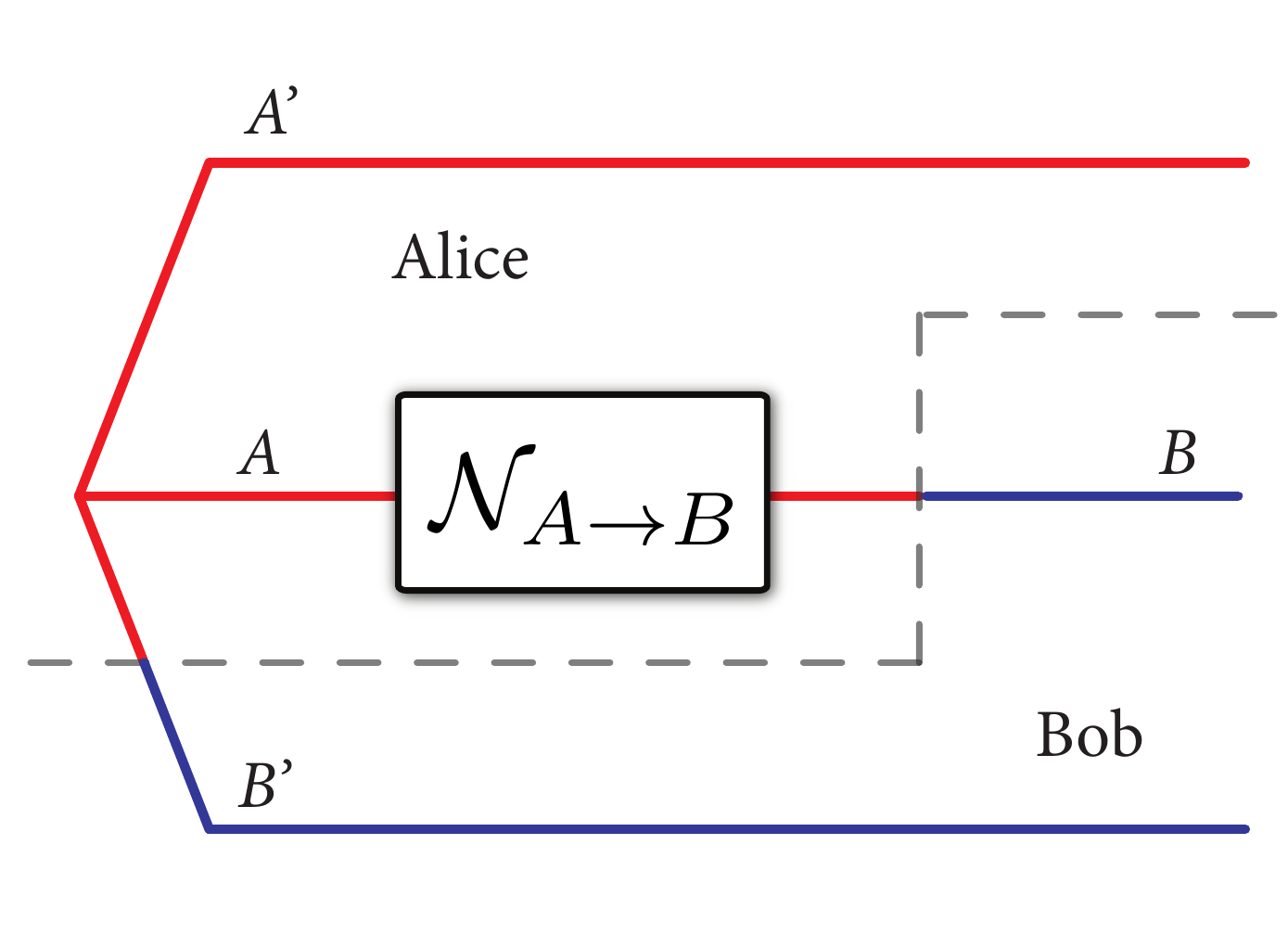}
\end{center}
\caption{The amortized entanglement of a quantum channel $\mathcal{N}_{A\to
B}$ is the largest difference in the entanglement between the output state on
systems $A^{\prime}:BB^{\prime}$ and the input state on systems $A^{\prime
}A:B^{\prime}$, the former of which is generated by the quantum channel
$\mathcal{N}_{A\rightarrow B}$.}%
\label{fig:amortized}%
\end{figure}

We begin by defining the amortized entanglement of a quantum channel as the
largest difference that can be achieved between the entanglement of an output
and input state of a quantum channel (see Figure~\ref{fig:amortized}\ for a
visual illustration of the scenario to which amortized entanglement
corresponds). Definition~\ref{def:amortized-ent} below applies to any
entanglement measure, which, as in \cite{H42007}, we define to be any function
of a bipartite quantum state that is monotone with respect to an LOCC channel,
i.e., a quantum channel that can be implemented by local operations and
classical communication (LOCC). As a minimal requirement, we also take an
entanglement measure to be equal to zero when evaluated on a product state and non-negative in general.

\begin{definition}
[Amortized entanglement of a quantum channel]\label{def:amortized-ent}For a
quantum channel $\mathcal{N}_{A\rightarrow B}$ and an entanglement measure
$E$, we define the channel's amortized entanglement as follows:%
\begin{equation}
E_{A}(\mathcal{N})\equiv\sup_{\rho_{A^{\prime}AB^{\prime}}}E(A^{\prime
};BB^{\prime})_{\theta}-E(A^{\prime}A;B^{\prime})_{\rho},
\end{equation}
where $\theta_{A^{\prime}BB^{\prime}}\equiv\mathcal{N}_{A\rightarrow B}%
(\rho_{A^{\prime}AB^{\prime}})$.
\end{definition}

As we stressed in the introduction, the quantity $E_{A}(\mathcal{N})$ is
closely related to ideas from prior work
\cite{BHLS03,LHL03,Christandl2017,BGMW17}, as well as the very recent
\cite{RKBKMA17}. Intuitively, the amortized entanglement of a channel captures
the largest difference in entanglement that can be generated between the
output and input of the channel.

Recall that the quantum relative entropy $D(\varsigma\Vert\xi)$ for a state
$\varsigma$ and a positive semi-definite operator $\xi$ is defined as
\cite{U62,Lindblad1973}%
\begin{equation}
D(\varsigma\Vert\xi)\equiv\operatorname{Tr}\{\varsigma\left[  \log
_{2}\varsigma-\log_{2}\xi\right]  \},
\end{equation}
whenever $\operatorname{supp}(\varsigma)\subseteq\operatorname{supp}(\xi)$ and
it is equal to $+\infty$ otherwise. We obtain two special cases of amortized
entanglement by considering the relative entropy of entanglement
\cite{VP98}\ and the Rains relative entropy \cite{R99,R01}\ as the underlying
entanglement measures. The former entanglement measure is relevant in the
context of secret-key distillation \cite{HHHO05,HHHO09}\ and the latter in the
context of entanglement distillation \cite{R99,R01}, both tasks performed with
respect to a bipartite state. Sections~\ref{sec:amortized-rel-ent-SKA} and
\ref{sec:amortized-Rains-q-comm} show how the amortized measures given below
are relevant in the context of secret-key-agreement and quantum communication
assisted by classical communication, respectively, both tasks performed with
respect to a quantum channel.

\begin{definition}
[Amortized relative entropy of entanglement]For a quantum channel
$\mathcal{N}_{A\rightarrow B}$, its amortized relative entropy of entanglement
is defined as follows:%
\begin{equation}
E_{AR}(\mathcal{N})\equiv\sup_{\rho_{A^{\prime}AB^{\prime}}}E_{R}(A^{\prime
};BB^{\prime})_{\theta}-E_{R}(A^{\prime}A;B^{\prime})_{\rho},
\end{equation}
where $\theta_{A^{\prime}BB^{\prime}}\equiv\mathcal{N}_{A\rightarrow B}%
(\rho_{A^{\prime}AB^{\prime}})$ and the relative entropy of entanglement
$E_{R}(C;D)_{\tau}$\ of a bipartite state $\tau_{CD}$ is defined as
\cite{VP98}%
\begin{equation}
E_{R}(C;D)_{\tau}\equiv\inf_{\sigma_{CD}\in\operatorname{SEP}(C:D)}D(\tau
_{CD}\Vert\sigma_{CD}),
\end{equation}
with $\operatorname{SEP}$ denoting the set of separable states \cite{W89}.
\end{definition}

\begin{definition}
[Amortized Rains relative entropy]\label{def:amortized-Rains}For a quantum
channel $\mathcal{N}_{A\rightarrow B}$, its amortized Rains relative entropy
is defined as follows:%
\begin{equation}
R_{A}(\mathcal{N})\equiv\sup_{\rho_{A^{\prime}AB^{\prime}}}R(A^{\prime
};BB^{\prime})_{\theta}-R(A^{\prime}A;B^{\prime})_{\rho},
\end{equation}
where $\theta_{A^{\prime}BB^{\prime}}\equiv\mathcal{N}_{A\rightarrow B}%
(\rho_{A^{\prime}AB^{\prime}})$ and the Rains relative entropy $R(C;D)_{\tau}%
$\ of a bipartite state $\tau_{CD}$ is defined as \cite{R99,R01}%
\begin{equation}
R(C;D)_{\tau}\equiv\inf_{\sigma_{CD}\in\operatorname{PPT}^{\prime}(C:D)}%
D(\tau_{CD}\Vert\sigma_{CD}),
\end{equation}
with $\operatorname{PPT}^{\prime}(C\!:\!D)$ denoting the Rains set
\cite{AdMVW02}:%
\begin{equation}
\operatorname{PPT}^{\prime}(C\!:\!D)=\{\sigma_{CD}:\sigma_{CD}\geq
0\wedge\left\Vert T_{D}(\sigma_{CD})\right\Vert _{1}\leq1\},
\end{equation}
and $T_{D}$ denotes the partial transpose of system $D$.
\end{definition}

Observe that $\operatorname{SEP}\subset\operatorname{PPT}^{\prime}$. Also,
note that the quantities $E_{AR}(\mathcal{N})$ and $R_{A}(\mathcal{N})$
involve an optimization over mixed states on systems $A^{\prime}AB^{\prime}$,
and we do not have an upper bound on the dimension of the $A^{\prime}$ or
$B^{\prime}$ systems. So these quantities could be difficult to calculate in
general. One of the main contributions of our paper (see
Section~\ref{sec:approx-TP-bound}) is to show how this quantity can be
approximated well in certain cases.

\subsection{Amortized entanglement versus the entanglement of a channel}

For any entanglement measure $E$, the entanglement of the channel is defined
as \cite{TGW14IEEE,TGW14Nat,TWW14,PLOB15,RKBKMA17}%
\begin{equation}
E(\mathcal{N})\equiv\sup_{\psi_{A^{\prime}A}}E(A^{\prime};B)_{\theta},
\end{equation}
where%
\begin{equation}
\theta_{A^{\prime}B}\equiv\mathcal{N}_{A\rightarrow B}(\psi_{A^{\prime}A}),
\end{equation}
and $\psi_{A^{\prime}A}$ is an arbitrary pure bipartite state with system $A'$ isomorphic to the channel input system~$A$. It suffices to optimize over pure states of the above form instead of general mixed states, due to purification, Schmidt decomposition, and monotonicity of the entanglement measure $E$ with respect to local operations (one of which is partial trace). Particular measures of interest are a channel's 
relative entropy of entanglement and the Rains relative entropy:
\begin{align}
E_{R}(\mathcal{N})  &  =\sup_{\psi_{A^{\prime}A}}E_{R}(A^{\prime}%
;B)_{\theta},\\
R(\mathcal{N})  &  =\sup_{\psi_{A^{\prime}A}}R(A^{\prime};B)_{\theta%
}.
\end{align}

The amortized entanglement of a channel is never smaller than that channel's
entanglement. That is, we always have the following inequality:%
\begin{equation}
E_{A}(\mathcal{N})\geq E(\mathcal{N}), \label{eq:amortized->=usual}%
\end{equation}
by taking $B^{\prime}$ to be a trivial system in
Definition~\ref{def:amortized-ent}.

The squashed entanglement $E_{\operatorname{sq}}$\ is a special entanglement
measure that obeys many desirable properties \cite{CW04,KW04,C06,BCY11,LW14}
(see also the various discussions in \cite{T99,T02}). One can also define the
dynamic version of this entanglement measure as the squashed entanglement of a
channel \cite{TGW14IEEE}, denoted as $E_{\operatorname{sq}}(\mathcal{N})$. A
particular property of squashed entanglement was established as
\cite[Theorem~7]{TGW14IEEE}. We remark here (briefly) that \cite[Theorem~7]%
{TGW14IEEE} implies the following inequality for the amortized version of
squashed entanglement%
\begin{equation}
E_{A,\operatorname{sq}}(\mathcal{N})\leq E_{\operatorname{sq}}(\mathcal{N}),
\end{equation}
which by \eqref{eq:amortized->=usual}, implies the following equality for
squashed entanglement:%
\begin{equation}
E_{A,\operatorname{sq}}(\mathcal{N})=E_{\operatorname{sq}}(\mathcal{N}).
\end{equation}
Thus, the squashed entanglement is rather special, in the sense that
amortization does not enhance its value.

\subsection{Convexity of a channel's amortized entaglement}

An entanglement measure $E$ is convex with respect to states \cite{H42007} if for all bipartite states $\rho^{0}_{CD}$ and $\rho^{1}_{CD}$ and $\lambda \in [0,1]$, the following equality holds:
\begin{equation}
E(C;D)_{\rho^\lambda} \leq 
\lambda E(C;D)_{\rho^0} + 
(1-\lambda) E(C;D)_{\rho^1},
\end{equation}
where $\rho^\lambda_{CD} = 
\lambda \rho^0_{CD} + 
(1-\lambda) \rho^1_{CD}$.
As the following proposition states, this property extends to amortized entanglement:

\begin{proposition}
[Convexity] Let $E$ be an  entanglement measure that is convex with respect to states. Then the amortized entanglement $E_{A}$\ of a channel is convex with respect to channels, in the sense that the following inequality holds for all quantum channels $\mathcal{N}^0$ and $\mathcal{N}^1$ and
$\lambda \in [0,1]$: 
\begin{equation}
E_{A}(\mathcal{N}^\lambda)\leq \lambda E_A(\mathcal{N}^0)+(1-\lambda)E_A(\mathcal{N}^1),
\end{equation}
where $\mathcal{N}^\lambda = \lambda \mathcal{N}^0+(1-\lambda)\mathcal{N}^1$.

\end{proposition}

\begin{proof}
Let $\rho_{A'AB'}$ be a state and set $\tau_{A'BB'}=\lambda \mathcal{N}^0_{A\rightarrow B}(\rho_{A'AB'})+(1-\lambda)\mathcal{N}^1(\rho_{A'AB'})$. Then consider that
\begin{align}
&E(A';BB')_{\tau}-E(A'A;B')_{\rho}
\nonumber
\\&= E(A';BB')_{\tau}-\lambda E(A'A;B')_{\rho}-(1-\lambda)E(A'A;B')_{\rho}  \\
&\leq \lambda E(A';BB')_{\mathcal{N}^0(\rho)}+(1-\lambda)E(A';BB')_{\mathcal{N}^1(\rho)}-\lambda E(A'A;B')_{\rho}-(1-\lambda)E(A'A;B')_{\rho}\\
&= \lambda (E(A';BB')_{\mathcal{N}^0(\rho)}-E(A'A;B')_{\rho})+(1-\lambda) (E(A';BB')_{\mathcal{N}^1(\rho)}-E(A'A;B')_{\rho})\\
&\leq \lambda  E_A(\mathcal{N}^0)+(1-\lambda)E_A(\mathcal{N}^1)\label{eq:convex_ineq}
\end{align}
The first equality follows from expanding the second term. The first inequality follows from the convexity of the entanglement measure $E$. The second equality from rearrangement of the terms. The second inequality follows from taking the supremum over all states $\rho_{A'AB'}$. Since the inequality in $\eqref{eq:convex_ineq}$ holds for all states $\rho_{A'AB'}$, the proof is complete.
\end{proof}

\subsection{Faithfulness of a channel's amortized entanglement}

An entanglement measure $E$ is faithful if it is equal to zero if and only if the state on which it is evaluated is a separable state. A quantum channel $\mathcal{N}$ is entanglement-breaking
\cite{HSR03}
if for all input states $\rho_{RA}$, the output state $(\operatorname{id}_R \otimes \mathcal{N}_{A\to B})(\rho_{RA})$ is a separable state. The following proposition extends the faithfulness property of entanglement measures to amortized entanglement and entanglement-breaking channels:

\begin{proposition}
[Faithfulness] Let $E$ be an  entanglement measure that is equal to zero for all separable states. If a channel $\mathcal{N}$ is entanglement-breaking, then its amortized entanglement $E_A(\mathcal{N})$ is equal to zero. If the entanglement measure $E$ is faithful and the amortized entanglement $E_{A}(\mathcal{N})$ of a channel $\mathcal{N}$ is equal to zero, then the channel $\mathcal{N}$ is entanglement-breaking.
\end{proposition}

\begin{proof}
We begin by proving the first statement above: for an arbitrary entanglement measure $E$, if a channel $\mathcal{N}$ is entanglement breaking, then its amortized entanglement is equal to zero. 
Let $\rho_{A'AB'}$ be an arbitrary input state to the channel, and let $\sigma_{A'BB'}=\mathcal{N}(\rho_{A'AB'})$ denote the output of the channel.
Recall that any entanglement-breaking channel can be represented as a measurement of the input system, followed by the preparation of a state on the output system, conditioned on the outcome of the measurement \cite{HSR03}. As such, the channel itself can be implemented by LOCC from the sender to the receiver. Then consider that
\begin{equation}
E(A';BB')_{\sigma}-E(A'A;B')_{\rho}
\leq E(A'A;B')_{\rho}-E(A'A;B')_{\rho}=0.
\end{equation}
The inequality follows from the fact that $E$ is an entanglement measure and is thus monotone with respect to LOCC. Given that we always have $E_{A}(\mathcal{N}) \geq 0$, we conclude that $E_{A}(\mathcal{N}) = 0$. 


Now we prove the second statement above: if the entanglement measure $E$ is faithful and the amortized entanglement $E_{A}(\mathcal{N})$ of a channel $\mathcal{N}$ is equal to zero, then the channel $\mathcal{N}$ is entanglement-breaking.  From \eqref{eq:amortized->=usual}, we have that $E_A(\mathcal{N})\geq E(\mathcal{N})$, which in turn implies that  $E(\mathcal{N})=0$. 
Since the underlying entanglement measure $E$ is faithful and
$E(A';B)_{\mathcal{N}(\rho)} \leq E(\mathcal{N})$ for all mixed input states $\rho_{A'A}$, we conclude that for all mixed input states  $ \rho_{A'A}$, the output state $\mathcal{N}_{A \to B}(\rho_{A'A})$ is separable. Thus, $\mathcal{N}$ is entanglement breaking.  
\end{proof}

\subsection{(Sub)additivity of a channel's amortized entanglement}

\begin{proposition}
[Subadditivity]\label{prop:amortized-subadditive}For any entanglement
measurement $E$, the amortized entanglement $E_{A}$\ of a channel is a
subadditive function of quantum channels, in the sense that the following
inequality holds for quantum channels $\mathcal{N}$ and $\mathcal{M}$:%
\begin{equation}
E_{A}(\mathcal{N}\otimes\mathcal{M})\leq E_{A}(\mathcal{N})+E_{A}%
(\mathcal{M}). \label{eq:amortized-subadd-gen}%
\end{equation}

\end{proposition}

\begin{proof}
Let $A_{1}$ and $B_{1}$ denote the respective input and output systems for
quantum channel$~\mathcal{N}$, and let $A_{2}$ and $B_{2}$ denote the
respective input and output quantum systems for quantum channel$~\mathcal{M}%
$.\ Let $\rho_{A^{\prime}A_{1}A_{2}B^{\prime}}$ denote a state to consider at
the input of $\mathcal{N}\otimes\mathcal{M}$, when optimizing the amortized
entanglement. Let $\theta_{A^{\prime}B_{1}B_{2}B^{\prime}}=(\mathcal{N}%
_{A_{1}\rightarrow B_{1}}\otimes\mathcal{M}_{A_{2}\rightarrow B_{2}}%
)(\rho_{A^{\prime}A_{1}A_{2}B^{\prime}})$, which is the state at the output of
the channel $\mathcal{N}\otimes\mathcal{M}$ when inputting $\rho_{A^{\prime
}A_{1}A_{2}B^{\prime}}$. Define the intermediary state $\tau_{A^{\prime}%
A_{1}B_{2}B^{\prime}}=\mathcal{M}_{A_{2}\rightarrow B_{2}}(\rho_{A^{\prime
}A_{1}A_{2}B^{\prime}})$. Then consider that%
\begin{align}
&  E(A^{\prime};B_{1}B_{2}B^{\prime})_{\theta}-E(A^{\prime}A_{1}%
A_{2};B^{\prime})_{\rho}\nonumber\\
&  =E(A^{\prime};B_{1}B_{2}B^{\prime})_{\theta}-E(A^{\prime}A_{1}%
;B_{2}B^{\prime})_{\tau}+E(A^{\prime}A_{1};B_{2}B^{\prime})_{\tau}%
-E(A^{\prime}A_{1}A_{2};B^{\prime})_{\rho}\\
&  \leq E_{A}(\mathcal{N})+E_{A}(\mathcal{M}).
\label{eq:amortized-subadd-critical-ineq}%
\end{align}
The first equality follows by adding and subtracting $E(A^{\prime}A_{1}%
;B_{2}B^{\prime})_{\tau}$. The second inequality follows because the states
$\tau_{A^{\prime}A_{1}B_{2}B^{\prime}}$ and $\theta_{A^{\prime}B_{1}%
B_{2}B^{\prime}}$ are particular states to consider at the respective input
and output for the amortized entanglement of the channel $\mathcal{N}$, by
making the identifications $A^{\prime}\leftrightarrow A^{\prime}$, $B^{\prime
}\leftrightarrow B^{\prime}B_{2}$, $B\leftrightarrow B_{1}$, and
$A\leftrightarrow A_{1}$, while the states $\rho_{A^{\prime}A_{1}%
A_{2}B^{\prime}}$\ and$\ \tau_{A^{\prime}A_{1}B_{2}B^{\prime}}$ are particular
states to consider at the respective input and output for the amortized
entanglement of the channel $\mathcal{M}$, by making the identifications
$A^{\prime}\leftrightarrow A^{\prime}A_{1}$, $B^{\prime}\leftrightarrow
B^{\prime}$, $B\leftrightarrow B_{2}$, and $A\leftrightarrow A_{2}$. Since the
inequality in \eqref{eq:amortized-subadd-critical-ineq} holds for all states
$\rho_{A^{\prime}A_{1}A_{2}B^{\prime}}$, we can conclude the inequality in \eqref{eq:amortized-subadd-gen}.
\end{proof}

\bigskip 

An immediate consequence of Proposition~\ref{prop:amortized-subadditive}\ is
the following inequality:%
\begin{equation}
\sup_{\mathcal{M}}\left[  E_{A}(\mathcal{N}\otimes\mathcal{M})-E_{A}%
(\mathcal{M})\right]  \leq E_{A}(\mathcal{N}),
\end{equation}
where the supremum is with respect to a quantum channel $\mathcal{M}$. This
inequality demonstrates that no other channel can help to enhance the
amortized entanglement of a quantum channel. See \cite{SSW08,WY16}\ for
related notions, i.e., potential capacity.

An entanglement measure $E$ is additive with respect to states \cite{H42007}%
\ if the following equality holds%
\begin{equation}
E(C_{1}C_{2};D_{1}D_{2})_{\tau}=E(C_{1};D_{1})_{\xi}+E(C_{2};D_{2})_{\zeta},
\end{equation}
where $\tau_{C_{1}C_{2}D_{1}D_{2}}=\xi_{C_{1}D_{1}}\otimes\zeta_{C_{2}D_{2}}$
and $\xi_{C_{1}D_{1}}$ and $\zeta_{C_{2}D_{2}}$ are bipartite states. It is
subadditive if%
\begin{equation}
E(C_{1}C_{2};D_{1}D_{2})_{\tau}\leq E(C_{1};D_{1})_{\xi}+E(C_{2};D_{2}%
)_{\zeta},
\end{equation}
and this latter property holds for both the relative entropy of entanglement
and the Rains relative entropy \cite{H42007}. The following proposition states
that amortized entanglement is additive if the underlying entanglement measure
is additive:

\begin{proposition}
[Additivity]\label{prop:amortized-additive}For any entanglement measurement
$E$ that is additive with respect to states, the amortized entanglement
$E_{A}$\ of a channel is an additive function of quantum channels, in the
sense that the following equality holds for quantum channels $\mathcal{N}$ and
$\mathcal{M}$:%
\begin{equation}
E_{A}(\mathcal{N}\otimes\mathcal{M})=E_{A}(\mathcal{N})+E_{A}(\mathcal{M}).
\label{eq:additivity-amortized}%
\end{equation}

\end{proposition}

\begin{proof}
The inequality $\leq$ holds for all channels as shown in
Proposition~\ref{prop:amortized-subadditive}. To see the other inequality, let
$\rho_{A_{1}^{\prime}A_{1}B_{1}^{\prime}}\otimes\kappa_{A_{2}^{\prime}%
A_{2}B_{2}^{\prime}}$ be an arbitrary state to consider for $E_{A}%
(\mathcal{N}\otimes\mathcal{M})$, and let%
\begin{equation}
\theta_{A_{1}^{\prime}B_{1}B_{1}^{\prime}A_{2}^{\prime}B_{2}B_{2}^{\prime}%
}=\mathcal{N}_{A_{1}\rightarrow B_{1}}(\rho_{A_{1}^{\prime}A_{1}B_{1}^{\prime
}})\otimes\mathcal{M}_{A_{2}\rightarrow B_{2}}(\kappa_{A_{2}^{\prime}%
A_{2}B_{2}^{\prime}}).
\end{equation}
Then%
\begin{align}
E_{A}(\mathcal{N}\otimes\mathcal{M})  &  \geq E(A_{1}^{\prime}A_{2}^{\prime
};B_{1}B_{2}B_{1}^{\prime}B_{2}^{\prime})_{\theta}-E(A_{1}^{\prime}%
A_{2}^{\prime}A_{1}A_{2};B_{1}^{\prime}B_{2}^{\prime})_{\rho\otimes\kappa}\\
&  =E(A_{1}^{\prime};B_{1}B_{1}^{\prime})_{\theta}-E(A_{1}^{\prime}A_{1}%
;B_{1}^{\prime})_{\rho}+E(A_{2}^{\prime};B_{2}B_{2}^{\prime})_{\theta}%
-E(A_{2}^{\prime}A_{2};B_{2}^{\prime})_{\kappa}%
\end{align}
The equality follows from the assumption that the underlying entanglement
measure is additive with respect to states. Since the above inequality holds
for all input states $\rho_{A_{1}^{\prime}A_{1}B_{1}^{\prime}}$ and
$\kappa_{A_{2}^{\prime}A_{2}B_{2}^{\prime}}$, we can conclude
\eqref{eq:additivity-amortized} after applying
Definition~\ref{def:amortized-ent}.
\end{proof}

\subsection{Amortized entanglement and teleportation simulation}

\label{sec:TP-simulation-def}Teleportation simulation of a quantum channel is
one of the earliest and most central insights in quantum information theory
\cite{BDSW96}, and it is a key tool used to establish upper bounds on
capacities of quantum channels assisted by local operations and classical
communication (LOCC) \cite{BDSW96,WPG07,NFC09,Mul12}. The basic idea behind
this tool is that a quantum channel can be simulated by the action of a
teleportation protocol \cite{PhysRevLett.70.1895,prl1998braunstein,Werner01}%
\ on a resource state $\omega_{RB}$\ shared between the sender $A$\ and
receiver $B$. More generally, a channel $\mathcal{N}_{A\rightarrow B}$ with
input system $A$ and output system $B$ is defined to be
teleportation-simulable with associated resource state $\omega_{RB}$ if the
following equality holds for all input states $\rho_{A}$ \cite[Eq.(11)]{HHH99}:%
\begin{equation}
\mathcal{N}_{A\rightarrow B}(\rho_{A})=\mathcal{L}_{ARB\rightarrow B}(\rho
_{A}\otimes\omega_{RB}),
\end{equation}
where $\mathcal{L}_{ARB\rightarrow B}$ is a quantum channel consisting of
LOCC\ between the sender, who has systems $A$ and $R$, and the receiver, who
has system $B$ ($\mathcal{L}_{ARB\rightarrow B}$ can also be considered a
generalized teleportation protocol, as in \cite{Werner01}).

Whenever the underlying entanglement measure is subadditive with respect to
quantum states, then one can easily bound the amortized entanglement
$E_{A}(\mathcal{N})$\ from above for channels that are teleportation-simulable:

\begin{proposition}
\label{prop:tp-upper-bound}Let $E_{S}$ be an entanglement measure that is
subadditive with respect to states, and let $E_{AS}$ denote its amortized
version. If a channel $\mathcal{N}_{A\rightarrow B}$\ is
teleportation-simulable with associated state $\omega_{RB}$, then the
following bound holds%
\begin{equation}
E_{AS}(\mathcal{N})\leq E_{S}(R;B)_{\omega},
\end{equation}
where $E_{AS}(\mathcal{N})$ denotes the amortized entanglement defined through
$E_{S}$ and Definition~\ref{def:amortized-ent}.
\end{proposition}

\begin{proof}
By the definition of a teleportation-simulable channel, we have that%
\begin{equation}
\mathcal{N}_{A\rightarrow B}(\rho_{A})=\mathcal{L}_{ARB\rightarrow B}(\rho
_{A}\otimes\omega_{RB}),
\end{equation}
where $\mathcal{L}_{ARB\rightarrow B}$ is an LOCC\ channel. Then for any input
state $\rho_{A^{\prime}B^{\prime}A}$, we have that%
\begin{align}
E_{S}(A^{\prime};BB^{\prime})_{\mathcal{L}(\rho\otimes\omega)}-E_{S}%
(A^{\prime}A;B^{\prime})_{\rho}  &  \leq E_{S}(A^{\prime}AR;B^{\prime}%
B)_{\rho\otimes\omega}-E_{S}(A^{\prime}A;B^{\prime})_{\rho}\\
&  \leq E_{S}(A^{\prime}A;B^{\prime})_{\rho}+E_{S}(R;B)_{\omega}%
-E_{S}(A^{\prime}A;B^{\prime})_{\rho}\\
&  =E_{S}(R;B)_{\omega}.
\end{align}
The first inequality follows from monotonicity of $E_{S}$ with respect to LOCC
channels (the fact that $E_{S}$ is an entanglement measure). The second
inequality follows from the assumption that $E_{S}$ is subadditive.
\end{proof}

Proposition~\ref{prop:tp-upper-bound}\ implies that the amortized entanglement
of a channel never exceeds the entanglement of the maximally entangled state,
whenever the underlying entanglement measure is subadditive. This follows
because any channel can be simulated by teleportation using the maximally
entangled state as the resource state, along with local processing. In
particular, Alice could apply the channel locally to her system and then
teleport it to Bob; also, she could first teleport to Bob and then he could
perform the local processing. So this leads to the following upper bound on
amortized entanglement in this case:

\begin{proposition}[Dimension bound]
\label{prop:tp-dim-upper-bound}Let $E_{S}$ be an entanglement measure that is
subadditive with respect to states, and let $E_{AS}$ denote its amortized
version. Let $\mathcal{N}_{A\rightarrow B}$ be a quantum channel. The
following bound holds%
\begin{equation}
E_{AS}(\mathcal{N})\leq\min\{E_{S}(A;\bar{A})_{\Phi},E_{S}(B;\bar{B})_{\Phi
}\},
\end{equation}
where $E_{AS}(\mathcal{N})$ denotes the amortized entanglement defined through
$E_{S}$ and Definition~\ref{def:amortized-ent}, $\bar{A}$ is a system
isomorphic to the channel input system $A$, $\bar{B}$ is a system isomorphic
to the channel output system $B$, and $\Phi$ denotes the maximally entangled
state. For the amortized relative entropy of entanglement and the amortized
Rains relative entropy, the above implies that%
\begin{equation}
E_{AR}(\mathcal{N}),\ R_{A}(\mathcal{N})\leq\log_{2}\min\{|A|,|B|\},
\end{equation}
because these underlying entanglement measures are equal to $\log_{2}d$ when
evaluated on a maximally entangled state of Schmidt rank $d$ \cite{H42007}.
\end{proposition}

In certain cases, the inequality in Proposition~\ref{prop:tp-upper-bound}\ is
actually an equality:

\begin{proposition}
\label{prop:TP-sim-amortized-equal}Let $E_{S}$ be an entanglement measure that
is subadditive with respect to states, and let $E_{AS}$ denote its amortized
version. If a channel $\mathcal{N}_{A\rightarrow B}$\ is
teleportation-simulable with associated state $\omega_{RB}=\mathcal{N}%
_{A\rightarrow B}(\rho_{RA})$ for some input state $\rho_{RA}$, then the
following equality holds%
\begin{equation}
E_{AS}(\mathcal{N})=E_{S}(R;B)_{\omega}.
\end{equation}

\end{proposition}

\begin{proof}
From Proposition~\ref{prop:tp-upper-bound}, we have that $E_{AS}%
(\mathcal{N})\leq E_{S}(R;B)_{\omega}$. The other inequality follows by
picking $\rho_{A^{\prime}B^{\prime}A}=\rho_{RA}$, with the identification
$A^{\prime}\leftrightarrow R$ and $B^{\prime}\leftrightarrow\emptyset$ (i.e.,
$B^{\prime}$ is a trivial system), and then we find that%
\begin{equation}
E_{AS}(\mathcal{N})=\sup_{\rho_{A^{\prime}AB^{\prime}}}E_{S}(A^{\prime
};BB^{\prime})_{\theta}-E_{S}(A^{\prime}A;B^{\prime})_{\rho}\geq
E_{S}(R;B)_{\omega}.
\end{equation}
This concludes the proof.
\end{proof}

\begin{remark}
For several channels with sufficient symmetry, such as covariant channels, one
can pick the input state $\rho_{RA}$ in
Proposition~\ref{prop:TP-sim-amortized-equal}\ to be the maximally entangled
state~$\Phi_{RA}$ \cite[Section~7]{CDP09}.
\end{remark}

\subsection{Uniform continuity of amortized relative entropy of entanglement
and amortized Rains relative entropy}

The following theorem establishes that both the amortized relative entropy of
entanglement and the amortized Rains relative entropy obey a uniform
continuity bound. This bound will play a central role in bounding the
respective secret-key-agreement and LOCC-assisted quantum capacities of
approximately teleportation-simulable channels (see
Sections~\ref{sec:amortized-rel-ent-SKA} and \ref{sec:amortized-Rains-q-comm}).

Before we state the theorem, we recall that the diamond norm of the difference
of two quantum channels $\mathcal{N}_{A\rightarrow B}$ and $\mathcal{M}%
_{A\rightarrow B}$ is defined as \cite{K97}%
\begin{align}
\left\Vert \mathcal{N}_{A\rightarrow B}-\mathcal{M}_{A\rightarrow
B}\right\Vert _{\Diamond}  &  =\sup_{\rho_{RA}}\left\Vert \left[
\operatorname{id}_{R}\otimes(\mathcal{N}_{A\rightarrow B}-\mathcal{M}%
_{A\rightarrow B})\right]  (\rho_{RA})\right\Vert _{1}\\
&  =\max_{\psi_{RA}}\left\Vert \left[  \operatorname{id}_{R}\otimes
(\mathcal{N}_{A\rightarrow B}-\mathcal{M}_{A\rightarrow B})\right]  (\psi
_{RA})\right\Vert _{1},
\end{align}
with $\left\Vert X\right\Vert _{1}=\operatorname{Tr}\{\sqrt{X^{\dag}X}\}$ and
the second equality, with an optimization restricted to pure states $\psi
_{RA}$ with $\left\vert R\right\vert =\left\vert A\right\vert $, follows from
the convexity of the trace norm and the Schmidt decomposition. The diamond
norm is a well established and operationally meaningful measure of the
distinguishability of two quantum channels.

\begin{theorem}
\label{thm:continuity}Let $\varepsilon\in\lbrack0,1]$. Let $E$ refer to either
the relative entropy of entanglement or the Rains relative entropy, and let
$E_{A}$ refer to their amortized versions. For channels $\mathcal{N}%
_{A\rightarrow B}$ and $\mathcal{M}_{A\rightarrow B}$ such that%
\begin{equation}
\frac{1}{2}\left\Vert \mathcal{N}_{A\rightarrow B}-\mathcal{M}_{A\rightarrow
B}\right\Vert _{\Diamond}\leq\varepsilon, \label{eq:diamond-norm-bound}%
\end{equation}
the following bound holds%
\begin{equation}
\left\vert E_{A}(\mathcal{N})-E_{A}(\mathcal{M})\right\vert \leq
2\varepsilon\log_{2}\left\vert B\right\vert +g(\varepsilon),
\end{equation}
where $\left\vert B\right\vert $ is the dimension of the channel output system
$B$ and $g(\varepsilon)\equiv\left(  \varepsilon+1\right)  \log_{2}%
(\varepsilon+1)-\varepsilon\log_{2}\varepsilon$.
\end{theorem}

\begin{proof}
Our proof follows the general approach from \cite{Winter15}, but it has some
additional observations needed for our context. For a state $\rho_{A^{\prime
}AB^{\prime}}$, let us define%
\begin{equation}
E_{A}(\rho,\mathcal{N})\equiv E(A^{\prime};BB^{\prime})_{\theta^{\mathcal{N}}%
}-E(A^{\prime}A;B^{\prime})_{\rho},
\end{equation}
where $\theta_{A^{\prime}BB^{\prime}}^{\mathcal{N}}\equiv\mathcal{N}%
_{A\rightarrow B}(\rho_{A^{\prime}AB^{\prime}})$. Then consider that%
\begin{equation}
\left\vert E_{A}(\rho,\mathcal{N})-E_{A}(\rho,\mathcal{M})\right\vert
=\left\vert E(A^{\prime};BB^{\prime})_{\theta^{\mathcal{N}}}-E(A^{\prime
};BB^{\prime})_{\theta^{\mathcal{M}}}\right\vert ,
\end{equation}
where $\theta_{A^{\prime}BB^{\prime}}^{\mathcal{M}}\equiv\mathcal{M}%
_{A\rightarrow B}(\rho_{A^{\prime}AB^{\prime}})$. Our intent now is to prove
that the following bound holds for all states $\rho_{A^{\prime}AB^{\prime}}$%
\begin{equation}
\left\vert E_{A}(\rho,\mathcal{N})-E_{A}(\rho,\mathcal{M})\right\vert
\leq2\varepsilon\log_{2}\left\vert B\right\vert +g(\varepsilon).
\label{eq:bound-for-each-state}%
\end{equation}
Since the bound in \eqref{eq:diamond-norm-bound} holds, we can conclude that%
\begin{equation}
\frac{1}{2}\left\Vert \theta_{A^{\prime}BB^{\prime}}^{\mathcal{N}}%
-\theta_{A^{\prime}BB^{\prime}}^{\mathcal{M}}\right\Vert _{1}\equiv
\varepsilon_{0}\leq\varepsilon.
\end{equation}
Let us suppose that $\varepsilon_{0}>0$. Otherwise, the bound in
\eqref{eq:bound-for-each-state}\ trivially holds. Let us define the states
$\Omega_{A^{\prime}BB^{\prime}}$ as%
\begin{equation}
\Omega_{A^{\prime}BB^{\prime}}=\frac{\left[  \theta_{A^{\prime}BB^{\prime}%
}^{\mathcal{N}}-\theta_{A^{\prime}BB^{\prime}}^{\mathcal{M}}\right]  _{+}%
}{\operatorname{Tr}\{\left[  \theta_{A^{\prime}BB^{\prime}}^{\mathcal{N}%
}-\theta_{A^{\prime}BB^{\prime}}^{\mathcal{M}}\right]  _{+}\}}=\frac
{1}{\varepsilon_{0}}\left[  \theta_{A^{\prime}BB^{\prime}}^{\mathcal{N}%
}-\theta_{A^{\prime}BB^{\prime}}^{\mathcal{M}}\right]  _{+},
\end{equation}
where $\left[  \cdot\right]  _{+}$ denotes the positive part of an operator. Note that the equality
$\operatorname{Tr}\{\left[  \theta_{A^{\prime}BB^{\prime}}^{\mathcal{N}%
}-\theta_{A^{\prime}BB^{\prime}}^{\mathcal{M}}\right]  _{+}\} = \varepsilon_{0}$ follows from the fact that
$\theta_{A^{\prime}BB^{\prime}}^{\mathcal{N}%
}$ and $\theta_{A^{\prime}BB^{\prime}}^{\mathcal{M}}$ are states (see \cite{W15book} for more details).
Since%
\begin{align}
\theta_{A^{\prime}BB^{\prime}}^{\mathcal{N}}  &  =\theta_{A^{\prime}%
BB^{\prime}}^{\mathcal{N}}-\theta_{A^{\prime}BB^{\prime}}^{\mathcal{M}}%
+\theta_{A^{\prime}BB^{\prime}}^{\mathcal{M}}\\
&  \leq\left[  \theta_{A^{\prime}BB^{\prime}}^{\mathcal{N}}-\theta_{A^{\prime
}BB^{\prime}}^{\mathcal{M}}\right]  _{+}+\theta_{A^{\prime}BB^{\prime}%
}^{\mathcal{M}}\\
&  =\left(  1+\varepsilon_{0}\right)  \left(  \frac{1}{1+\varepsilon_{0}%
}\left[  \theta_{A^{\prime}BB^{\prime}}^{\mathcal{N}}-\theta_{A^{\prime
}BB^{\prime}}^{\mathcal{M}}\right]  _{+}+\frac{1}{1+\varepsilon_{0}}%
\theta_{A^{\prime}BB^{\prime}}^{\mathcal{M}}\right) \\
&  =\left(  1+\varepsilon_{0}\right)  \left(  \frac{\varepsilon_{0}%
}{1+\varepsilon_{0}}\Omega_{A^{\prime}BB^{\prime}}+\frac{1}{1+\varepsilon_{0}%
}\theta_{A^{\prime}BB^{\prime}}^{\mathcal{M}}\right)  ,
\end{align}
we can define%
\begin{equation}
\xi_{A^{\prime}BB^{\prime}}\equiv\frac{\varepsilon_{0}}{1+\varepsilon_{0}%
}\Omega_{A^{\prime}BB^{\prime}}+\frac{1}{1+\varepsilon_{0}}\theta_{A^{\prime
}BB^{\prime}}^{\mathcal{M}},
\end{equation}
and it follows that%
\begin{equation}
\xi_{A^{\prime}BB^{\prime}}=\frac{\varepsilon_{0}}{1+\varepsilon_{0}}%
\Omega_{A^{\prime}BB^{\prime}}^{\prime}+\frac{1}{1+\varepsilon_{0}}%
\theta_{A^{\prime}BB^{\prime}}^{\mathcal{N}},
\end{equation}
where the state $\Omega_{A^{\prime}BB^{\prime}}^{\prime}$ is defined as%
\begin{equation}
\Omega_{A^{\prime}BB^{\prime}}^{\prime}\equiv\frac{1}{\varepsilon_{0}}\left[
\left(  1+\varepsilon_{0}\right)  \xi_{A^{\prime}BB^{\prime}}-\theta
_{A^{\prime}BB^{\prime}}^{\mathcal{N}}\right]  .
\end{equation}
Consider that%
\begin{equation}
\Omega_{A^{\prime}B^{\prime}}=\Omega_{A^{\prime}B^{\prime}}^{\prime}
\label{eq:critical-omega}%
\end{equation}
because%
\begin{align}
\operatorname{Tr}_{B}\left\{  \frac{\varepsilon_{0}}{1+\varepsilon_{0}}%
\Omega_{A^{\prime}BB^{\prime}}+\frac{1}{1+\varepsilon_{0}}\theta_{A^{\prime
}BB^{\prime}}^{\mathcal{M}}\right\}   &  =\frac{\varepsilon_{0}}%
{1+\varepsilon_{0}}\Omega_{A^{\prime}B^{\prime}}+\frac{1}{1+\varepsilon_{0}%
}\theta_{A^{\prime}B^{\prime}}^{\mathcal{M}}\\
&  =\frac{\varepsilon_{0}}{1+\varepsilon_{0}}\Omega_{A^{\prime}B^{\prime}%
}+\frac{1}{1+\varepsilon_{0}}\rho_{A^{\prime}B^{\prime}},\\
\operatorname{Tr}_{B}\left\{  \frac{\varepsilon_{0}}{1+\varepsilon_{0}}%
\Omega_{A^{\prime}BB^{\prime}}^{\prime}+\frac{1}{1+\varepsilon_{0}}%
\theta_{A^{\prime}BB^{\prime}}^{\mathcal{N}}\right\}   &  =\frac
{\varepsilon_{0}}{1+\varepsilon_{0}}\Omega_{A^{\prime}B^{\prime}}^{\prime
}+\frac{1}{1+\varepsilon_{0}}\theta_{A^{\prime}B^{\prime}}^{\mathcal{N}}\\
&  =\frac{\varepsilon_{0}}{1+\varepsilon_{0}}\Omega_{A^{\prime}B^{\prime}%
}^{\prime}+\frac{1}{1+\varepsilon_{0}}\rho_{A^{\prime}B^{\prime}},
\end{align}
and so%
\begin{equation}
\frac{\varepsilon_{0}}{1+\varepsilon_{0}}\Omega_{A^{\prime}B^{\prime}}%
+\frac{1}{1+\varepsilon_{0}}\rho_{A^{\prime}B^{\prime}}=\frac{\varepsilon_{0}%
}{1+\varepsilon_{0}}\Omega_{A^{\prime}B^{\prime}}^{\prime}+\frac
{1}{1+\varepsilon_{0}}\rho_{A^{\prime}B^{\prime}},
\end{equation}
from which we can conclude \eqref{eq:critical-omega} since $\varepsilon_{0}%
>0$. By convexity of relative entropy of entanglement and the Rains relative
entropy \cite{H42007}, we have that%
\begin{equation}
E(A^{\prime};BB^{\prime})_{\xi}\leq\frac{1}{1+\varepsilon_{0}}E(A^{\prime
};BB^{\prime})_{\theta^{\mathcal{M}}}+\frac{\varepsilon_{0}}{1+\varepsilon
_{0}}E(A^{\prime};BB^{\prime})_{\Omega},
\end{equation}
and from Lemma~\ref{lem:rel-ent-mixture-ent-bound}, we have that%
\begin{equation}
\frac{1}{1+\varepsilon_{0}}E(A^{\prime};BB^{\prime})_{\theta^{\mathcal{N}}%
}+\frac{\varepsilon_{0}}{1+\varepsilon_{0}}E(A^{\prime};BB^{\prime}%
)_{\Omega^{\prime}}\leq E(A^{\prime};BB^{\prime})_{\xi}+h_{2}\!\left(
\frac{\varepsilon_{0}}{1+\varepsilon_{0}}\right)  ,
\end{equation}
So this means that%
\begin{align}
E(A^{\prime};BB^{\prime})_{\theta^{\mathcal{N}}}  &  \leq\left(
1+\varepsilon_{0}\right)  E(A^{\prime};BB^{\prime})_{\xi}+\left(
1+\varepsilon_{0}\right)  h_{2}\!\left(  \frac{\varepsilon_{0}}{1+\varepsilon
_{0}}\right)  -\varepsilon_{0}E(A^{\prime};BB^{\prime})_{\Omega^{\prime}}\\
&  \leq E(A^{\prime};BB^{\prime})_{\theta^{\mathcal{M}}}+\varepsilon
_{0}E(A^{\prime};BB^{\prime})_{\Omega}+g(\varepsilon_{0})-\varepsilon
_{0}E(A^{\prime};BB^{\prime})_{\Omega^{\prime}}\\
&  =E(A^{\prime};BB^{\prime})_{\theta^{\mathcal{M}}}+g(\varepsilon
_{0})+\varepsilon_{0}\left[  E(A^{\prime};BB^{\prime})_{\Omega}-E(A^{\prime
};BB^{\prime})_{\Omega^{\prime}}\right]  ,
\end{align}
where we used that $\left(  1+\varepsilon_{0}\right)  h_{2}\!\left(
\frac{\varepsilon_{0}}{1+\varepsilon_{0}}\right)  =g(\varepsilon_{0})$.
Finally, consider that%
\begin{align}
E(A^{\prime};BB^{\prime})_{\Omega}-E(A^{\prime};BB^{\prime})_{\Omega^{\prime
}}  &  \leq2\log_{2}\left\vert B\right\vert +E(A^{\prime};B^{\prime})_{\Omega
}-E(A^{\prime};B^{\prime})_{\Omega^{\prime}}\\
&  =2\log_{2}\left\vert B\right\vert .
\end{align}
The first inequality follows from Lemma~\ref{lem:dim-bound}, the
LOCC\ monotonicity of relative entropy of entanglement and the Rains relative
entropy, and \eqref{eq:critical-omega}. So this implies that%
\begin{align}
E(A^{\prime};BB^{\prime})_{\theta^{\mathcal{N}}}-E(A^{\prime};BB^{\prime
})_{\theta^{\mathcal{M}}}  &  \leq g(\varepsilon_{0})+2\varepsilon_{0}\log
_{2}\left\vert B\right\vert \\
&  \leq g(\varepsilon)+2\varepsilon\log_{2}\left\vert B\right\vert ,
\end{align}
the latter inequality holding because $g(\cdot)$ is a monotone increasing
function. The other inequality%
\begin{equation}
E(A^{\prime};BB^{\prime})_{\theta^{\mathcal{M}}}-E(A^{\prime};BB^{\prime
})_{\theta^{\mathcal{N}}}\leq g(\varepsilon)+2\varepsilon\log_{2}\left\vert
B\right\vert
\end{equation}
follows immediately using similar steps, and this now establishes
\eqref{eq:bound-for-each-state}. Then we have that the following bound holds
for all input states $\rho_{A^{\prime}B^{\prime}A}$:%
\begin{align}
E_{A}(\rho,\mathcal{N})  &  \leq\sup_{\rho}E_{A}(\rho,\mathcal{M}%
)+2\varepsilon\log_{2}\left\vert B\right\vert +g(\varepsilon)\\
&  =E_{A}(\mathcal{M})+2\varepsilon\log_{2}\left\vert B\right\vert
+g(\varepsilon).
\end{align}
Since the bound holds for all input states $\rho_{A^{\prime}B^{\prime}A}$, we
can conclude that%
\begin{equation}
E_{A}(\mathcal{N})\leq E_{A}(\mathcal{M})+2\varepsilon\log_{2}\left\vert
B\right\vert +g(\varepsilon).
\end{equation}
In a similar way, we obtain the opposite inequality%
\begin{equation}
E_{A}(\mathcal{M})\leq E_{A}(\mathcal{N})+2\varepsilon\log_{2}\left\vert
B\right\vert +g(\varepsilon),
\end{equation}
and this completes the proof.
\end{proof}

\section{Amortized relative entropy of entanglement and secret key agreement}

\label{sec:amortized-rel-ent-SKA}In this section, we prove that the amortized
relative entropy of entanglement is an upper bound on the secret-key-agreement
capacity of a quantum channel. We begin by reviewing the structure of a
secret-key-agreement protocol \cite{TGW14IEEE,TGW14Nat}, how such a protocol
can be purified along the lines observed in \cite{HHHO05,HHHO09}, the critical
performance parameters for such a protocol, and then we finally give a proof
for the aforementioned claim. Note that the proof bears some similarities with
proofs in prior works \cite{WTB16,Christandl2017,BGMW17}, as well as an argument that
appeared recently in \cite{DW17} in a different context.

\subsection{Protocol for secret key agreement}

\label{sec:SKA-protocol}Here we review the structure of a secret-key-agreement
protocol, along the lines discussed in \cite{TGW14IEEE,TGW14Nat}:

A sender Alice and a receiver Bob are spatially separated and are connected by
a quantum channel $\mathcal{N}_{A\rightarrow B}$. They begin by performing an
LOCC channel $\mathcal{L}_{\emptyset\rightarrow A_{1}^{\prime}A_{1}%
B_{1}^{\prime}}^{(1)}$, which leads to a separable state $\rho_{A_{1}^{\prime
}A_{1}B_{1}^{\prime}}^{(1)}$, where $A_{1}^{\prime}$ and $B_{1}^{\prime}$ are
systems that are finite-dimensional but arbitrarily large and $A_{1}$ is a
system that can be fed into the first channel use. Alice transmits system
$A_{1}$ into the first channel, leading to a state $\sigma_{A_{1}^{\prime
}B_{1}B_{1}^{\prime}}^{(1)}\equiv\mathcal{N}_{A_{1}\rightarrow B_{1}}%
(\rho_{A_{1}^{\prime}A_{1}B_{1}^{\prime}}^{(1)})$. They then perform the LOCC
channel $\mathcal{L}_{A_{1}^{\prime}B_{1}B_{1}^{\prime}\rightarrow
A_{2}^{\prime}A_{2}B_{2}^{\prime}}^{(2)}$, which leads to the state%
\begin{equation}
\rho_{A_{2}^{\prime}A_{2}B_{2}^{\prime}}^{(2)}\equiv\mathcal{L}_{A_{1}%
^{\prime}B_{1}B_{1}^{\prime}\rightarrow A_{2}^{\prime}A_{2}B_{2}^{\prime}%
}^{(2)}(\sigma_{A_{1}^{\prime}B_{1}B_{1}^{\prime}}^{(1)}).
\end{equation}
Alice feeds in the system $A_{2}$ to the second channel use $\mathcal{N}%
_{A_{2}\rightarrow B_{2}}$, leading to the state $\sigma_{A_{2}^{\prime}%
B_{2}B_{2}^{\prime}}^{(2)}\equiv\mathcal{N}_{A_{2}\rightarrow B_{2}}%
(\rho_{A_{2}^{\prime}A_{2}B_{2}^{\prime}}^{(1)})$. This process
continues:\ the protocol uses the channel $n$ times. In general, we have the
following states for all $i\in\{2,\ldots,n\}$:%
\begin{align}
\rho_{A_{i}^{\prime}A_{i}B_{i}^{\prime}}^{(i)}  &  \equiv\mathcal{L}%
_{A_{i-1}^{\prime}B_{i-1}B_{i-1}^{\prime}\rightarrow A_{i}^{\prime}A_{i}%
B_{i}^{\prime}}^{(i)}(\sigma_{A_{i-1}^{\prime}B_{i-1}B_{i-1}^{\prime}}%
^{(i-1)}),\\
\sigma_{A_{i}^{\prime}B_{i}B_{i}^{\prime}}^{(i)}  &  \equiv\mathcal{N}%
_{A_{i}\rightarrow B_{i}}(\rho_{A_{i}^{\prime}A_{i}B_{i}^{\prime}}^{(i)}),
\end{align}
where $\mathcal{L}_{A_{i-1}^{\prime}B_{i-1}B_{i-1}^{\prime}\rightarrow
A_{i}^{\prime}A_{i}B_{i}^{\prime}}^{(i)}$ is an LOCC\ channel. The final step
of the protocol consists of an LOCC channel $\mathcal{L}_{A_{n}^{\prime}%
B_{n}B_{n}^{\prime}\rightarrow K_{A}K_{B}}^{(n+1)}$, which produces the key
systems $K_{A}$ and $K_{B}$ for Alice and Bob, respectively. The final state
of the protocol is then as follows:%
\begin{equation}
\omega_{K_{A}K_{B}}\equiv\mathcal{L}_{A_{n}^{\prime}B_{n}B_{n}^{\prime
}\rightarrow K_{A}K_{B}}^{(n+1)}(\sigma_{A_{n}^{\prime}B_{n}B_{n}^{\prime}%
}^{(n)}).
\end{equation}

The goal of the protocol is that the final state $\omega_{K_{A}K_{B}}$ is
close to a secret-key state. Figure~\ref{fig:private-code}\ depicts such a
protocol. It may not yet be clear exactly what we mean by \textquotedblleft
close to a secret-key state,\textquotedblright\ but our approach is standard
and we clarify this point in the following two sections.

\begin{figure}[ptb]
\begin{center}
\includegraphics[
width=6.5638in
]{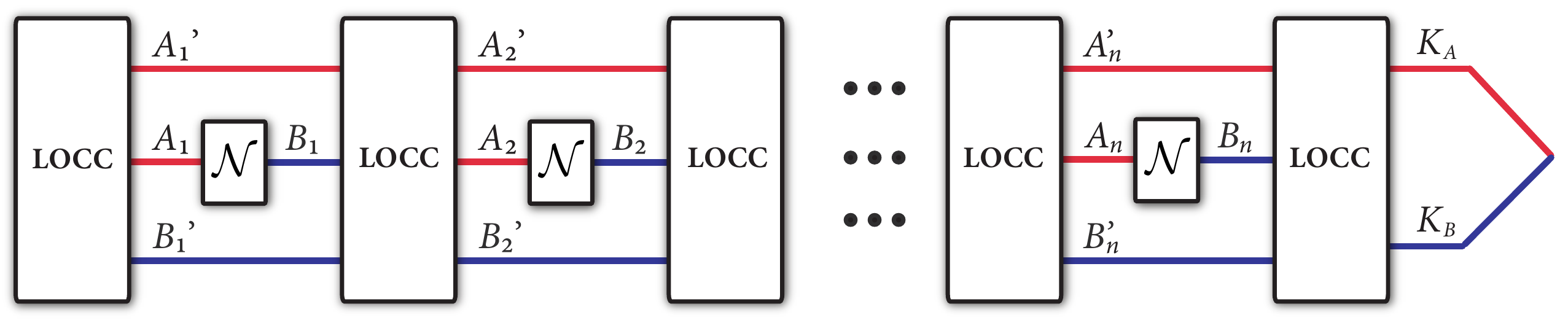}
\end{center}
\caption{A protocol for secret key agreement over a quantum channel.}%
\label{fig:private-code}%
\end{figure}

\subsection{Purifying a secret-key-agreement protocol}

Related to the observations in \cite{HHHO05,HHHO09}, any protocol of the above
form can be purified in the following sense. The initial state $\rho
_{A_{1}^{\prime}A_{1}B_{1}^{\prime}}^{(1)}$ is a separable state of the
following form:%
\begin{equation}
\rho_{A_{1}^{\prime}A_{1}B_{1}^{\prime}}^{(1)}\equiv\sum_{y_{1}}p_{Y_{1}%
}(y_{1})\tau_{A_{1}^{\prime}A_{1}}^{y_{1}}\otimes\zeta_{B_{1}}^{y_{1}}.
\end{equation}
The classical random variable $Y_{1}$ corresponds to a message exchanged
between Alice and Bob to establish this state. It can be purified in the
following way:%
\begin{equation}
|\psi^{(1)}\rangle_{A_{1}^{\prime}A_{1}S_{A_{1}}B_{1}^{\prime}S_{B_{1}}Y_{1}%
}\equiv\sum_{y_{1}}\sqrt{p_{Y_{1}}(y_{1})}|\tau^{y_{1}}\rangle_{A_{1}^{\prime
}A_{1}S_{A_{1}}}\otimes|\zeta^{y_{1}}\rangle_{B_{1}S_{B_{1}}}\otimes
|y_{1}\rangle_{Y_{1}},
\end{equation}
where $S_{A_{1}}$ and $S_{B_{1}}$ are local \textquotedblleft
shield\textquotedblright\ systems that in principle could be held by Alice and
Bob, respectively, $|\tau^{y_{1}}\rangle_{A_{1}^{\prime}A_{1}S_{A_{1}}}$ and
$|\zeta^{y_{1}}\rangle_{B_{1}S_{B_{1}}}$ purify $\tau_{A_{1}^{\prime}A_{1}%
}^{y_{1}}$ and $\zeta_{B_{1}}^{y_{1}}$, respectively, and Eve possesses system
$Y_{1}$, which contains a coherent classical copy of the classical data
exchanged. Each LOCC\ channel $\mathcal{L}_{A_{i-1}^{\prime}B_{i-1}%
B_{i-1}^{\prime}\rightarrow A_{i}^{\prime}A_{i}B_{i}^{\prime}}^{(i)}$\ can be
written in the following form \cite{Wat15}, for all $i\in\{2,\ldots,n\}$:%
\begin{equation}
\mathcal{L}_{A_{i-1}^{\prime}B_{i-1}B_{i-1}^{\prime}\rightarrow A_{i}^{\prime
}A_{i}B_{i}^{\prime}}^{(i)}\equiv\sum_{y_{i}}\mathcal{E}_{A_{i-1}^{\prime
}\rightarrow A_{i}^{\prime}A_{i}}^{y_{i}}\otimes\mathcal{F}_{B_{i-1}%
B_{i-1}^{\prime}\rightarrow B_{i}^{\prime}}^{y_{i}},
\label{eq:LOCC-as-separable}%
\end{equation}
where $\{\mathcal{E}_{A_{i-1}^{\prime}\rightarrow A_{i}^{\prime}A_{i}}^{y_{i}%
}\}_{y_{i}}$ and $\{\mathcal{F}_{B_{i-1}B_{i-1}^{\prime}\rightarrow
B_{i}^{\prime}}^{y_{i}}\}_{y_{i}}$ are collections of completely positive,
trace non-increasing maps such that the map in \eqref{eq:LOCC-as-separable} is
trace preserving. Such an LOCC\ channel can be purified to an isometry in the
following way:%
\begin{equation}
U_{A_{i-1}^{\prime}B_{i-1}B_{i-1}^{\prime}\rightarrow A_{i}^{\prime}%
A_{i}S_{A_{i}}B_{i}^{\prime}S_{B_{i}}Y_{i}}^{\mathcal{L}^{(i)}}\equiv
\sum_{y_{i}}U_{A_{i-1}^{\prime}\rightarrow A_{i}^{\prime}A_{i}S_{A_{i}}%
}^{\mathcal{E}^{y_{i}}}\otimes U_{B_{i-1}B_{i-1}^{\prime}\rightarrow
B_{i}^{\prime}S_{B_{i}}}^{\mathcal{F}^{y_{i}}}\otimes|y_{i}\rangle_{Y_{i}},
\label{eq:LOCC-as-separable-iso-ext}%
\end{equation}
where $\{U_{A_{i-1}^{\prime}\rightarrow A_{i}^{\prime}A_{i}S_{A_{i}}%
}^{\mathcal{E}^{y_{i}}}\}_{y_{i}}$ and $\{U_{B_{i-1}B_{i-1}^{\prime
}\rightarrow B_{i}^{\prime}S_{B_{i}}}^{\mathcal{F}^{y_{i}}}\}_{y_{i}}$ are
collections of linear operators (each of which is a contraction, i.e.,
$\left\Vert U_{A_{i-1}^{\prime}\rightarrow A_{i}^{\prime}A_{i}S_{A_{i}}%
}^{\mathcal{E}^{y_{i}}}\right\Vert _{\infty},\left\Vert U_{B_{i-1}%
B_{i-1}^{\prime}\rightarrow B_{i}^{\prime}S_{B_{i}}}^{\mathcal{F}^{y_{i}}%
}\right\Vert _{\infty}\leq1$) such that the linear operator in
\eqref{eq:LOCC-as-separable-iso-ext} is an isometry, and $Y_{i}$ is a system
containing a coherent classical copy of the classical data exchanged in this
round, the system $Y_{i}$\ being held by Eve. The final LOCC\ channel can be
written similarly as%
\begin{equation}
\mathcal{L}_{A_{n}^{\prime}B_{n}B_{n}^{\prime}\rightarrow K_{A}K_{B}}%
^{(n+1)}\equiv\sum_{y_{n+1}}\mathcal{E}_{A_{n}^{\prime}\rightarrow K_{A}%
}^{y_{n+1}}\otimes\mathcal{F}_{B_{n}B_{n}^{\prime}\rightarrow K_{B}}^{y_{n+1}%
},
\end{equation}
and it can be purified to an isometry similarly as%
\begin{equation}
U_{A_{n}^{\prime}B_{n}B_{n}^{\prime}\rightarrow K_{A}S_{A_{n+1}}%
K_{B}S_{B_{n+1}}}^{\mathcal{L}^{(n+1)}}\equiv\sum_{y_{n+1}}U_{A_{n}^{\prime
}\rightarrow K_{A}S_{A_{n+1}}}^{\mathcal{E}^{y_{n+1}}}\otimes U_{B_{n}%
B_{n}^{\prime}\rightarrow K_{B}S_{B_{n+1}}}^{\mathcal{F}^{y_{n+1}}}%
\otimes|y_{n+1}\rangle_{Y_{n+1}}.
\end{equation}
Furthermore, each channel use $\mathcal{N}_{A_{i}\rightarrow B_{i}}$, for all
$i\in\{1,\ldots,n\}$ is purified by an isometry $U_{A_{i}\rightarrow
B_{i}E_{i}}^{\mathcal{N}}$, such that Eve possesses the environment system
$E_{i}$.

\subsection{Performance of a secret-key-agreement protocol}

\label{sec:secret-key-performance}At the end of the purified protocol, Alice
possesses the key system $K_{A}$ and the shield systems $S_{A}\equiv S_{A_{1}%
}\cdots S_{A_{n+1}}$, Bob possesses the key system $K_{B}$ and the shield
systems $S_{B}\equiv S_{B_{1}}\cdots S_{B_{n+1}}$, and Eve possesses the
environment systems $E^{n}\equiv E_{1}\cdots E_{n}$ as well as the coherent
copies $Y^{n+1}\equiv Y_{1}\cdots Y_{n+1}$ of the classical data exchanged.
The state at the end of the purified protocol is a pure state $|\omega
\rangle_{K_{A}S_{A}K_{B}S_{B}E^{n}Y^{n+1}}$. Fix $n,K\in\mathbb{N}$ and
$\varepsilon\in\lbrack0,1]$. The original protocol is an $(n,K,\varepsilon)$
protocol if%
\begin{equation}
F(\omega_{K_{A}K_{B}E^{n}Y^{n+1}},\overline{\Phi}_{K_{A}K_{B}}\otimes
\xi_{E^{n}Y^{n+1}})\geq1-\varepsilon, \label{eq:fidelity-tripartite-setting}%
\end{equation}
where the fidelity $F(\tau,\kappa)\equiv\left\Vert \sqrt{\tau}\sqrt{\kappa
}\right\Vert _{1}^{2}$ \cite{U76}, the maximally correlated state
$\overline{\Phi}_{K_{A}K_{B}}$ is defined as%
\begin{equation}
\overline{\Phi}_{K_{A}K_{B}}\equiv\frac{1}{K}\sum_{k=1}^{K}|k\rangle\langle
k|_{K_{A}}\otimes|k\rangle\langle k|_{K_{B}},
\end{equation}
and $\xi_{E^{n}Y^{n+1}}$ is an arbitrary state.

By the observations of \cite{HHHO05,HHHO09}\ (understood as a clever
application of Uhlmann's theorem for fidelity \cite{U76}), rather than
focusing on the tripartite scenario, one can focus on the bipartite scenario,
in which the goal is to produce an approximate private state of Alice and
Bob's systems. The criterion in \eqref{eq:fidelity-tripartite-setting} is
fully equivalent to%
\begin{equation}
F(\omega_{K_{A}S_{A}K_{B}S_{B}},\gamma_{K_{A}S_{A}K_{B}S_{B}})\geq
1-\varepsilon,
\end{equation}
where $\gamma_{K_{A}S_{A}K_{B}S_{B}}$ is a private state \cite{HHHO05,HHHO09}
of the following form:%
\begin{equation}
U_{K_{A}S_{A}K_{B}S_{B}}(\Phi_{K_{A}K_{B}}\otimes\theta_{S_{A}S_{B}}%
)U_{K_{A}S_{A}K_{B}S_{B}}^{\dag},
\end{equation}
with $U_{K_{A}S_{A}K_{B}S_{B}}$ a twisting unitary of the form $U_{K_{A}%
S_{A}K_{B}S_{B}}=\sum_{i,j=1}^{K}|i\rangle\langle i|_{K_{A}}\otimes
|j\rangle\langle j|_{K_{B}}\otimes U_{S_{A}S_{B}}^{ij}$, $\Phi_{K_{A}K_{B}}$
is a maximally entangled state of the form%
\begin{equation}
\Phi_{K_{A}K_{B}}\equiv\frac{1}{K}\sum_{i,j=1}^{K}|i\rangle\langle j|_{K_{A}%
}\otimes|i\rangle\langle j|_{K_{B}}, \label{eq:MES}%
\end{equation}
and $\theta_{S_{A}S_{B}}$ is an arbitrary state of the shield systems. The
main idea in this latter picture is that we can use the techniques of
entanglement theory to understand private communication protocols
\cite{HHHO05,HHHO09}.

A rate $R$ is achievable for secret key agreement if for all $\varepsilon
\in(0,1)$, $\delta>0$, and sufficiently large$~n$, there exists an
$(n,2^{n\left(  R-\delta\right)  },\varepsilon)$ protocol. The
secret-key-agreement capacity of $\mathcal{N}_{A\rightarrow B}$, denoted as
$P^{\leftrightarrow}(\mathcal{N}_{A\rightarrow B})$, is equal to the supremum
of all achievable rates.

\subsection{Teleportation-simulable channels and reduction by teleportation}

\label{sec:TP-sim-reduct}An implication of channel simulation via
teleportation, as discussed in Section~\ref{sec:TP-simulation-def}, is that
the performance of a general protocol that uses the channel $n$ times, with
each use interleaved by local operations and classical communication (LOCC),
can be bounded from above by the performance of a protocol with a much simpler
form:\ the simplified protocol consists of a single round of LOCC acting on
$n$ copies of the resource state $\omega_{RB}$ \cite{BDSW96,NFC09,Mul12}. This
is called reduction by teleportation. Note that reduction by teleportation is
a very general procedure and clearly can be used more generally in any
LOCC-assisted protocol trying to accomplish an arbitrary
information-processing task. Of course, a secret-key-agreement protocol is one
particular kind of protocol of the above form, as considered in the follow-up
works \cite{PLOB15,WTB16}, and so the general reduction method of
\cite{BDSW96,NFC09,Mul12}\ applies to this particular case.

\subsection{Amortized relative entropy of entanglement as a bound for
secret-key-agreement protocols}

The main goal of this section is to show that the amortized relative entropy
of entanglement is an upper bound on the rate of secret key that can be
extracted by a secret-key-agreement protocol. We begin by establishing the
following theorem:

\begin{proposition}
\label{prop:weak-converse-bound}The following weak-converse bound holds for an
$(n,K,\varepsilon)$ secret-key-agreement protocol conducted over a quantum
channel $\mathcal{N}$:%
\begin{equation}
(1-\varepsilon)\frac{\log_{2}K}{n}\leq E_{AR}(\mathcal{N})+\frac{1}{n}%
h_{2}(\varepsilon), \label{eq:amortized-weak-converse-bound}%
\end{equation}
where $E_{AR}(\mathcal{N})$ is the amortized relative entropy of entanglement and $h_{2}(\varepsilon)\equiv-\varepsilon\log_{2}\varepsilon-\left(1-\varepsilon
\right)\log_{2}\!\left(1-\varepsilon\right)$ denotes the binary entropy.
\end{proposition}

\begin{proof}
To see how the amortized relative entropy of entanglement gives an upper bound
on the performance of a secret-key-agreement protocol, consider the following
steps. We suppose that we are dealing with an $(n,K,\varepsilon)$
secret-key-agreement protocol as described previously. First, at the end of
the protocol one can perform a privacy test \cite{WTB16}\ (see also
\cite{HHHO09,HHHLO08,PhysRevLett.100.110502}), which untwists the twisting
unitary of the ideal target private state and projects onto the maximally
entangled state of the key systems (cf.,
Section~\ref{sec:secret-key-performance}). Let $F^{\ast}$ denote the
probability of the actual state $\omega_{K_{A}S_{A}K_{B}S_{B}}$\ at the end of
the protocol passing this test. By \cite[Lemma~9]{WTB16}, this probability is
larger than $1-\varepsilon$. Let $p_{\operatorname{SEP}}$ denote the
probability that a given separable state $\sigma_{K_{A}S_{A}K_{B}S_{B}}$\ of
systems $K_{A}S_{A}K_{B}S_{B}$ passes the test. This probability is no larger
than $1/K$, following from results of \cite{HHHO09} (reviewed as
\cite[Lemma~10]{WTB16}). Then we find that%
\begin{align}
-h_{2}(\varepsilon)+(1-\varepsilon)\log_{2}K  &  \leq D(\{1-\varepsilon
,\varepsilon\}\Vert\{1/K,1-1/K\})\\
&  \leq D(\{F^{\ast},1-F^{\ast}\}\Vert\{p_{\operatorname{SEP}}%
,1-p_{\operatorname{SEP}}\})\\
&  \leq D(\omega_{K_{A}S_{A}K_{B}S_{B}}\Vert\sigma_{K_{A}S_{A}K_{B}S_{B}}).
\end{align}
The first inequality follows because%
\begin{align}
D(\{1-\varepsilon,\varepsilon\}\Vert\{1/K,1-1/K\})  &  =(1-\varepsilon
)\log_{2}\left(  \frac{1-\varepsilon}{1/K}\right)  +\varepsilon\log_{2}\left(
\frac{\varepsilon}{1-1/K}\right) \\
&  =-h_{2}(\varepsilon)+\left(  1-\varepsilon\right)  \log_{2}K+\varepsilon
\log_{2}\left(  \frac{K}{K-1}\right) \\
&  \geq-h_{2}(\varepsilon)+\left(  1-\varepsilon\right)  \log_{2}K.
\end{align}
with the last inequality above following because $K\geq1$ (note that the
singular case of $K=1$ is not particularly interesting because the rate is
zero and the protocol is thus trivial). The second inequality follows because
the distributions $\{F^{\ast},1-F^{\ast}\}$ and $\{p_{\operatorname{SEP}%
},1-p_{\operatorname{SEP}}\}$ are more distinguishable than the distributions
$\{1-\varepsilon,\varepsilon\}$ and $\{1/K,1-1/K\}$, due to the conditions
$F^{\ast}\geq1-\varepsilon$ and $p_{\operatorname{SEP}}\leq1/K$ (we are
assuming without loss of generality that $1-\varepsilon\geq1/K$ for this
statement; if it is not the case, then the code is in a rather sad state of
affairs with poor performance and there is no need to give a converse bound in
this case---the bound would simply be $\log_{2}K<-\log_{2}(1-\varepsilon)$).
The final inequality follows from monotonicity of quantum relative entropy
with respect to the privacy test (understood as a measurement channel). Since
the above chain of inequalities holds for all separable states $\sigma
_{K_{A}S_{A}K_{B}S_{B}}$, we find that%
\begin{equation}
-h_{2}(\varepsilon)+(1-\varepsilon)\log_{2}K\leq E_{R}(K_{A}S_{A};K_{B}%
S_{B})_{\omega},
\end{equation}
which can be rewritten as\footnote{Alternatively, the bound in
\eqref{eq:first-bound-final-state} can be established for all $\varepsilon
\in(0,1)$ and $K\geq1$ by using several methods from
\cite{WR12,MW12,WTB16,KW17}. In more detail, we could make use of the
$\varepsilon$-relative entropy of entanglement \cite{BD11}\ as an intermediary
step to get that $\log_{2}K\leq E_{R}^{\varepsilon}(K_{A}S_{A};K_{B}%
S_{B})_{\omega}\leq\frac{1}{1-\varepsilon}\left[  E_{R}(K_{A}S_{A};K_{B}%
S_{B})_{\omega}+h_{2}(\varepsilon)\right]  $, with the first bound following
from \cite[Theorem~11]{WTB16} and the second from \cite{WR12,MW12,KW17}.}%
\begin{equation}
(1-\varepsilon)\log_{2}K\leq E_{R}(K_{A}S_{A};K_{B}S_{B})_{\omega}%
+h_{2}(\varepsilon). \label{eq:first-bound-final-state}%
\end{equation}
From the monotonicity of the relative entropy of entanglement with respect to
LOCC \cite{VP98}, we find that%
\begin{align}
E_{R}(K_{A}S_{A};K_{B}S_{B})_{\omega}  &  \leq E_{R}(A_{n}^{\prime};B_{n}%
B_{n}^{\prime})_{\sigma^{(n)}}\\
&  =E_{R}(A_{n}^{\prime};B_{n}B_{n}^{\prime})_{\sigma^{(n)}}-E_{R}%
(A_{1}^{\prime}A_{1};B_{1}^{\prime})_{\rho^{(1)}}\\
&  =E_{R}(A_{n}^{\prime};B_{n}B_{n}^{\prime})_{\sigma^{(n)}}+\left[
\sum_{i=2}^{n}E_{R}(A_{i}^{\prime}A_{i};B_{i}^{\prime})_{\rho^{(i)}}%
-E_{R}(A_{i}^{\prime}A_{i};B_{i}^{\prime})_{\rho^{(i)}}\right] \nonumber\\
&  \qquad-E_{R}(A_{1}^{\prime}A_{1};B_{1}^{\prime})_{\rho^{(1)}}\\
&  \leq\sum_{i=1}^{n}E_{R}(A_{i}^{\prime};B_{i}B_{i}^{\prime})_{\sigma^{(i)}%
}-E_{R}(A_{i}^{\prime}A_{i};B_{i}^{\prime})_{\rho^{(i)}}\\
&  \leq nE_{AR}(\mathcal{N}). \label{eq:last-bound-amortized-proof}%
\end{align}
The first equality follows because the state $\rho_{A_{1}^{\prime}A_{1}%
B_{1}^{\prime}}^{(1)}$ is a separable state with vanishing relative entropy of
entanglement. The second equality follows by adding and subtracting terms. The
second inequality follows because $E_{R}(A_{i}^{\prime}A_{i};B_{i}^{\prime
})_{\rho^{(i)}}\leq E_{R}(A_{i-1}^{\prime};B_{i-1}B_{i-1}^{\prime}%
)_{\sigma^{(i-1)}}$ for all $i\in\{2,\ldots,n\}$, due to monotonicity of the
relative entropy of entanglement with respect to LOCC. The final inequality
follows because each term $E_{R}(A_{n}^{\prime};B_{n}B_{n}^{\prime}%
)_{\sigma^{(i)}}-E_{R}(A_{i}^{\prime}A_{i};B_{i}^{\prime})_{\rho^{(i)}}$ is of
the form in the amortized relative entropy of entanglement, so that optimizing
over all inputs of the form $\rho^{(i)}$\ cannot exceed $E_{AR}(\mathcal{N})$.
Combining \eqref{eq:first-bound-final-state}\ and
\eqref{eq:last-bound-amortized-proof}, we arrive at the inequality in \eqref{eq:amortized-weak-converse-bound}.
\end{proof}

\bigskip

Taking the limit in Proposition~\ref{prop:weak-converse-bound}\ as
$n\rightarrow\infty$ and then as $\varepsilon\rightarrow0$ leads to the
following asymptotic statement:

\begin{theorem}
\label{thm:amortized-rel-ent-bound}The secret-key-agreement capacity
$P^{\leftrightarrow}(\mathcal{N}_{A\rightarrow B})$ of a quantum channel
$\mathcal{N}_{A\rightarrow B}$ cannot exceed its amortized relative entropy of
entanglement:%
\begin{equation}
P^{\leftrightarrow}(\mathcal{N})\leq E_{AR}(\mathcal{N}).
\end{equation}

\end{theorem}

\begin{remark}
Interestingly, a similar approach using the sandwiched R\'{e}nyi relative
entropy \cite{MDSFT13,WWY13}\ gives the following upper bound for all
$\alpha>1$%
\begin{equation}
\frac{\alpha}{\alpha-1}\log_{2}(1-\varepsilon)\leq n\widetilde{E}_{AR}%
^{\alpha}(\mathcal{N})-\log_{2}K, \label{eq:alpha-rel-ent-bound-1}%
\end{equation}
where%
\begin{equation}
\widetilde{E}_{AR}^{\alpha}(\mathcal{N})=\sup_{\rho_{A^{\prime}AB^{\prime}}%
}\widetilde{E}_{R}^{\alpha}(A^{\prime};BB^{\prime})_{\omega}-\widetilde{E}%
_{R}^{\alpha}(A^{\prime}A;B^{\prime})_{\rho}, \label{eq:alpha-rel-ent-bound-2}%
\end{equation}
and $\widetilde{E}_{R}^{\alpha}$ denotes the sandwiched R\'{e}nyi relative
entropy of entanglement \cite{WTB16},\ defined from the sandwiched R\'{e}nyi
relative entropy \cite{MDSFT13,WWY13}. One of the main results of
\cite{Christandl2017} is the bound $\widetilde{E}_{AR}^{\alpha}(\mathcal{N}%
)\leq E_{\max}(\mathcal{N})$, where $E_{\max}(\mathcal{N})$ denotes the
channel's max-relative entropy of entanglement (cf., \cite{D09}).\ The authors
of \cite{Christandl2017} observed that this latter bound in turn implies the
following bound for all $\alpha>1$:%
\begin{equation}
\frac{\alpha}{\alpha-1}\log_{2}(1-\varepsilon)\leq nE_{\max}(\mathcal{N}%
)-\log_{2}K.
\end{equation}

\end{remark}

\section{Amortized Rains relative entropy and PPT-assisted quantum
communication}

\label{sec:amortized-Rains-q-comm}The main goal of this section is to prove
that the amortized Rains relative entropy from
Definition~\ref{def:amortized-Rains} is an upper bound on a channel's
PPT-assisted quantum capacity. By this, we mean that a sender and receiver are
allowed to use a channel many times, and between every channel use, they are
allowed free usage of channels that are positive-partial-transpose
(PPT)\ preserving. In what follows, we detail these concepts, and then we
state the main theorem (Theorem~\ref{thm:amortized-rel-ent-bound-PPT}).

\subsection{Positive-partial-transpose preserving quantum channels}

A quantum channel $\mathcal{P}$ is a positive-partial-transpose (PPT)
preserving channel from systems $A\!:\!B$ to systems $A^{\prime}%
\!:\!B^{\prime}$ if the map $T_{B^{\prime}}\circ\mathcal{N}_{AB\rightarrow
A^{\prime}B^{\prime}}\circ T_{B}$ is completely positive and trace
preserving~\cite{R01}, where $T_{B}$ and $T_{B^{\prime}}$ denote the partial
transpose map. For a given basis $\{|i\rangle\}_{i}$, the transpose map is a
positive map, specified by $\rho\rightarrow\sum_{i,j}|i\rangle\langle
j|\rho|i\rangle\langle j|$. In what follows, we call PPT-preserving channels
\textquotedblleft PPT\ channels\textquotedblright\ as an abbreviation. It has
been known for a long time that PPT\ channels contain the set of LOCC channels
\cite{R01}, and so an immediate operational consequence of this containment is
that any general upper bound on the performance of a PPT-assisted protocol
serves as an upper bound on the performance of an LOCC-assisted protocol
\cite{R01}. This fact and the fact that PPT\ channels are simpler to analyze
mathematically than LOCC were some of the main motivations for introducing
this class of channels \cite{R01}.

PPT\ channels preserve the set $\operatorname{PPT}^{\prime}$ discussed in
Definition~\ref{def:amortized-Rains} \cite{R01,AdMVW02}. For this reason and
since the relative entropy is monotone with respect to quantum channels
\cite{Lindblad1975}, it follows that the Rains relative entropy is monotone
with respect to PPT\ channels \cite{R01,AdMVW02}, in the sense that%
\begin{equation}
R(A;B)_{\rho}\geq R(A^{\prime};B^{\prime})_{\mathcal{P}(\rho)},
\end{equation}
where $\rho_{AB}$ is a bipartite state and $\mathcal{P}_{AB\rightarrow
A^{\prime}B^{\prime}}$ is a PPT\ channel.

The notion of PPT\ channels then leads to a more general notion of the
teleportation simulation of a quantum channel:

\begin{definition}
[$\omega$-PPT-simulable channel]\label{def:PPT-sim-chan}A channel
$\mathcal{N}_{A\rightarrow B}$ with input system $A$ and output system $B$ is
defined to be PPT-simulable with associated resource state $\omega_{RB}$
($\omega$-PPT-simulable for short) if the following equality holds for all
input states $\rho_{A}$:%
\begin{equation}
\mathcal{N}_{A\rightarrow B}(\rho_{A})=\mathcal{P}_{ARB\rightarrow B}(\rho
_{A}\otimes\omega_{RB}),
\end{equation}
where $\mathcal{P}_{ARB\rightarrow B}$ is a PPT quantum channel with respect
to the bipartite cut $AR|B$ at the input. Note that every
teleportation-simulable channel with associated resource state $\omega_{RB}$
is PPT-simulable with associated resource state $\omega_{RB}$.
\end{definition}

\subsection{Protocols for PPT-assisted quantum communication and their
performance}

\label{sec:PPT-protocols}The structure of an $(n,M,\varepsilon)$ protocol for
PPT-assisted quantum communication is quite similar to that for an
$(n,K,\varepsilon)$ protocol for secret key agreement, which we discussed
previously in Section~\ref{sec:SKA-protocol}. In fact, such a PPT-assisted
protocol is exactly as outlined in Section~\ref{sec:SKA-protocol} and
Figure~\ref{fig:private-code}, but each LOCC\ channel is replaced with a
PPT\ channel. Let us denote the final state of the protocol by $\omega
_{M_{A}M_{B}}$ instead of $\omega_{K_{A}K_{B}}$. Fixing $n,M\in\mathbb{N}$ and
$\varepsilon\in\lbrack0,1]$, the protocol is an $(n,M,\varepsilon)$
PPT-assisted quantum communication protocol if%
\begin{equation}
F(\omega_{M_{A}M_{B}},\Phi_{M_{A}M_{B}})\geq1-\varepsilon,
\end{equation}
where the maximally entangled state $\Phi_{M_{A}M_{B}}$ is defined in \eqref{eq:MES}.

A rate $R$ is achievable for PPT-assisted quantum communication if for all
$\varepsilon\in(0,1)$, $\delta>0$, and sufficiently large $n$, there exists an
$(n,2^{n\left(  R-\delta\right)  },\varepsilon)$ protocol. The PPT-assisted
quantum capacity of $\mathcal{N}_{A\rightarrow B}$, denoted as
$Q^{\mathbf{PPT},\leftrightarrow}(\mathcal{N}_{A\rightarrow B})$, is equal to
the supremum of all achievable rates.

We can also consider the whole development above when we only allow the
assistance of LOCC\ channels instead of PPT\ channels. In this case, we have
similar notions as above, and then we arrive at the LOCC-assisted quantum
capacity $Q^{\leftrightarrow}(\mathcal{N}_{A\rightarrow B})$. It then
immediately follows that%
\begin{equation}
Q^{\leftrightarrow}(\mathcal{N}_{A\rightarrow B})\leq Q^{\mathbf{PPT}%
,\leftrightarrow}(\mathcal{N}_{A\rightarrow B})
\end{equation}
because every LOCC\ channel is a PPT\ channel. We also have the following
bound%
\begin{equation}
Q^{\leftrightarrow}(\mathcal{N}_{A\rightarrow B})\leq P^{\leftrightarrow
}(\mathcal{N}_{A\rightarrow B}),
\end{equation}
as observed in \cite{TWW14}, because a maximally entangled state, the target
state of an LOCC-assisted quantum communication protocol, is a particular kind
of private state.

For channels that are $\omega$-PPT-simulable, as in
Definition~\ref{def:PPT-sim-chan}, PPT-assisted protocols simplify immensely,
just as was the case for teleportation-simulable channels. Indeed, in this
case, PPT-assisted protocols can be reduced to the action of a single
PPT\ channel on $n$ copies of the resource state $\omega_{RB}$, and this
reduction is helpful in bounding the performance of PPT-assisted protocols
conducted over such channels.

\subsection{Amortized Rains relative entropy as a bound for PPT-assisted
quantum communication protocols}

We can employ an argument nearly identical to that given in the proof of
Proposition~\ref{prop:weak-converse-bound}\ in order to establish that the
amortized Rains relative entropy is an upper bound on the rate at which
maximal entanglement can be extracted by a PPT-assisted quantum communication
protocol. Indeed, we simply replace the privacy test therein by a maximal
entanglement test (i.e., a measurement specified by a projection onto the
maximally entangled state or its complement). By the definition of an
$(n,M,\varepsilon)$ protocol and the fidelity, the probability for the final
state $\omega_{M_{A}M_{B}}$ of the protocol to pass this test is larger than
$1-\varepsilon$. Furthermore, due to \cite[Lemma~2]{R99}, the bound
$\operatorname{Tr}\{\Phi_{M_{A}M_{B}}\sigma_{M_{A}M_{B}}\}\leq1/M$ holds for
all states $\sigma_{M_{A}M_{B}}\in\operatorname{PPT}^{\prime}(M_{A}%
\!:\!M_{B})$. These bounds and the same reasoning employed in the proof of
Proposition~\ref{prop:weak-converse-bound}\ allow us to conclude the following
weak-converse bound for any PPT-assisted quantum communication protocol:

\begin{proposition}
\label{prop:weak-converse-bound-PPT}The following weak-converse bound holds
for an $(n,M,\varepsilon)$ PPT-assisted quantum communication protocol
conducted over a quantum channel $\mathcal{N}$:%
\begin{equation}
(1-\varepsilon)\frac{\log_{2}M}{n}\leq R_{A}(\mathcal{N})+\frac{1}{n}%
h_{2}(\varepsilon),
\end{equation}
where $R_{A}(\mathcal{N})$ denotes the amortized Rains relative entropy from
Definition~\ref{def:amortized-Rains}.
\end{proposition}

Taking the limit in Proposition~\ref{prop:weak-converse-bound-PPT}\ as
$n\rightarrow\infty$ and then as $\varepsilon\rightarrow0$ leads to the
following asymptotic statement:

\begin{theorem}
\label{thm:amortized-rel-ent-bound-PPT}The PPT-assisted quantum capacity
$Q^{\mathbf{PPT},\leftrightarrow}(\mathcal{N}_{A\rightarrow B})$ of a quantum
channel $\mathcal{N}_{A\rightarrow B}$ cannot exceed its amortized Rains
relative entropy:%
\begin{equation}
Q^{\mathbf{PPT},\leftrightarrow}(\mathcal{N})\leq R_{A}(\mathcal{N}).
\end{equation}

\end{theorem}

Furthermore, we obtain a bound for PPT-assisted quantum communication similar
to that stated in
\eqref{eq:alpha-rel-ent-bound-1}--\eqref{eq:alpha-rel-ent-bound-2}\ by
employing similar reasoning.

\section{Approximately teleportation- and PPT-simulable channels}

\label{sec:approx-TP-bound}

We now define approximately teleportation- and PPT-simulable channels:

\begin{definition}
[Approximately teleportation- and PPT-simulable channels]%
\label{def:approx-TP-sim}A quantum channel $\mathcal{N}_{A\rightarrow B}$ is
$\varepsilon$-approximately teleportation-simulable with associated resource
state $\omega_{RB}$ if there exists a channel $\mathcal{M}_{A\rightarrow B}$
that is exactly teleportation-simulable with associated resource state
$\omega_{RB}$ such that%
\begin{equation}
\frac{1}{2}\left\Vert \mathcal{N}_{A\rightarrow B}-\mathcal{M}_{A\rightarrow
B}\right\Vert _{\Diamond}\leq\varepsilon.
\end{equation}
For short, we say that $\mathcal{N}_{A\rightarrow B}$ is $\left(
\varepsilon,\omega_{RB}\right)  $-approximately teleportation-simulable. The
same definition applies for an $\left(  \varepsilon,\omega_{RB}\right)
$-approximately PPT-simulable channel, but the difference is that
$\mathcal{M}_{A\rightarrow B}$ is exactly PPT-simulable with associated
resource state $\omega_{RB}$. Also, if a channel is $\left(  \varepsilon
,\omega_{RB}\right)  $-approximately teleportation-simulable, then it is
$\left(  \varepsilon,\omega_{RB}\right)  $-approximately PPT-simulable.
\end{definition}

In Appendix~\ref{sec:cov-TP-sim}, we discuss the relation between the notion of
approximately teleportation-simulable channels and the recently introduced
notion of approximately covariant channels \cite{LKDW17}. Therein, we also
discuss channel twirling and how to simulate this procedure via a generalized
teleportation protocol.

The following theorem is an immediate consequence of
Proposition~\ref{prop:tp-upper-bound}, Theorem~\ref{thm:continuity}, and
Theorem~\ref{thm:amortized-rel-ent-bound}, and it constitutes one of the main
results of our paper:

\begin{theorem}
\label{thm:approx-TP-sim-bound}If a channel $\mathcal{N}_{A\rightarrow B}$ is
$\left(  \varepsilon,\omega_{RB}\right)  $-approximately
teleportation-simulable, then its secret-key-agreement capacity
$P^{\leftrightarrow}(\mathcal{N}_{A\rightarrow B})$ is bounded from above as%
\begin{equation}
P^{\leftrightarrow}(\mathcal{N}_{A\rightarrow B})\leq E_{R}(R;B)_{\omega
}+2\varepsilon\log_{2}\left\vert B\right\vert +g(\varepsilon).
\end{equation}

\end{theorem}

Similarly, the following theorem is an immediate consequence of
Proposition~\ref{prop:tp-upper-bound}, Theorem~\ref{thm:continuity}, and
Theorem~\ref{thm:amortized-rel-ent-bound-PPT}:

\begin{theorem}
\label{thm:approx-TP-sim-bound-PPT}If a channel $\mathcal{N}_{A\rightarrow B}$
is $\left(  \varepsilon,\omega_{RB}\right)  $-approximately PPT-simulable,
then its PPT-assisted quantum capacity $Q^{\mathbf{PPT},\leftrightarrow
}(\mathcal{N}_{A\rightarrow B})$ is bounded from above as%
\begin{equation}
Q^{\mathbf{PPT},\leftrightarrow}(\mathcal{N}_{A\rightarrow B})\leq
R(R;B)_{\omega}+2\varepsilon\log_{2}\left\vert B\right\vert +g(\varepsilon).
\end{equation}

\end{theorem}

In the next section, we apply the bounds from
Theorems~\ref{thm:approx-TP-sim-bound}\ and \ref{thm:approx-TP-sim-bound-PPT}%
\ to an example qubit channel.

\section{Bounds on the assisted capacities of a particular qubit channel}

\label{sec:examples}

In this section, we apply the bounds from
Theorems~\ref{thm:approx-TP-sim-bound}\ and \ref{thm:approx-TP-sim-bound-PPT}%
\ to a particular qubit channel $\mathcal{N}_{p}$, which we define to be a
convex mixture of an amplitude damping channel and a depolarizing channel:%
\begin{equation}
\mathcal{N}_{p}(\rho)=p\mathcal{A}_{p}(\rho)+(1-p)\mathcal{D}_{p}(\rho),
\label{eqn:channel}%
\end{equation}
where $p\in\lbrack0,1]$ and $\rho$ is an input qubit density operator (this is
the same channel considered in concurrent work \cite{LKDW17}). The amplitude
damping channel $\mathcal{A}_{p}$ is defined as
\begin{align}
\mathcal{A}_{p}(\rho)  &  =K_{1}\rho K_{1}^{\dagger}+K_{2}\rho K_{2}^{\dagger
},\\
K_{1}  &  =|0\rangle\langle0|+\sqrt{1-p}|1\rangle\langle1|,\\
K_{2}  &  =\sqrt{p}|0\rangle\langle1|.
\end{align}
Also, $\mathcal{D}_{p}$ denotes the qubit depolarizing channel:%
\begin{equation}
\mathcal{D}_{p}(\rho)=(1-p)\rho+\frac{p}{3}(X\rho X+Y\rho Y+Z\rho Z),
\end{equation}
where $X$, $Y$, and $Z$ are the Pauli operators.

Let $\Phi(\mathcal{M})$ denote the Choi state associated with a channel
$\mathcal{M}$, i.e., the state that results from sending one share of a
maximally entangled state through the channel. It has been known for many
years now \cite{BDSW96,HHH99,CDP09}\ that the depolarizing channel is
teleportation-simulable with associated resource state $\Phi(\mathcal{D}_{p}%
)$. Thus, for small values of $p$, we should expect for a convex mixture of
the depolarizing channel and the amplitude channel to be approximately
teleportation-simulable with associated resource state $\Phi(\mathcal{D}_{p}%
)$, given that $\mathcal{N}_{p}$ is intuitively close to $\mathcal{D}_{p}$ for
small values of $p$. Indeed, it follows from the results of \cite[Section~4.2]{LKDW17} and
the discussion in Appendix~\ref{sec:cov-TP-sim} that the channel
$\mathcal{N}_{p}$ is $(p^{2}/2,\Phi(\overline{\mathcal{N}}_{p}))$%
-approximately teleportation-simulable, where $\overline{\mathcal{N}}_{p}$
denotes the following teleportation-simulable channel:%
\begin{equation}
\overline{\mathcal{N}}_{p}(\rho)=\frac{1}{2}\left[  \mathcal{N}_{p}%
(\rho)+X\mathcal{N}_{p}(X\rho X)X\right]  .
\end{equation}
We can thus apply Theorems~\ref{thm:approx-TP-sim-bound} and
\ref{thm:approx-TP-sim-bound-PPT} to arrive at the following bounds on the
secret-key-agreement capacity $P^{\leftrightarrow}(\mathcal{N}_{p})$ and the
PPT-assisted quantum capacity $Q^{\mathbf{PPT},\leftrightarrow}(\mathcal{N}%
_{p})$:%
\begin{align}
P^{\leftrightarrow}(\mathcal{N}_{p}) &  \leq E_{R}(A;B)_{\Phi(\overline
{\mathcal{N}}_{p})}+p^{2}+g(p^{2}/2),\\
Q^{\mathbf{PPT},\leftrightarrow}(\mathcal{N}_{p}) &  \leq R(A;B)_{\Phi
(\overline{\mathcal{N}}_{p})}+p^{2}+g(p^{2}/2).
\end{align}
By noting that the set $\operatorname{PPT}^{\prime}$ contains the set of PPT
states and applying the well known result that PPT\ states are equal to
separable states for $2\times2$ systems \cite{P96,Horodecki19961}, we can
conclude that%
\begin{equation}
P^{\leftrightarrow}(\mathcal{N}_{p}),\ Q^{\mathbf{PPT},\leftrightarrow
}(\mathcal{N}_{p})\leq E_{\mathbf{PPT}}(A;B)_{\Phi(\overline{\mathcal{N}}%
_{p})}+p^{2}+g(p^{2}/2),\label{eq:qubit-chan-upper-bound}%
\end{equation}
where $E_{\mathbf{PPT}}$ denotes the relative entropy to PPT\ states.

\begin{figure}[ptb]
\begin{center}
\includegraphics[width=10cm]{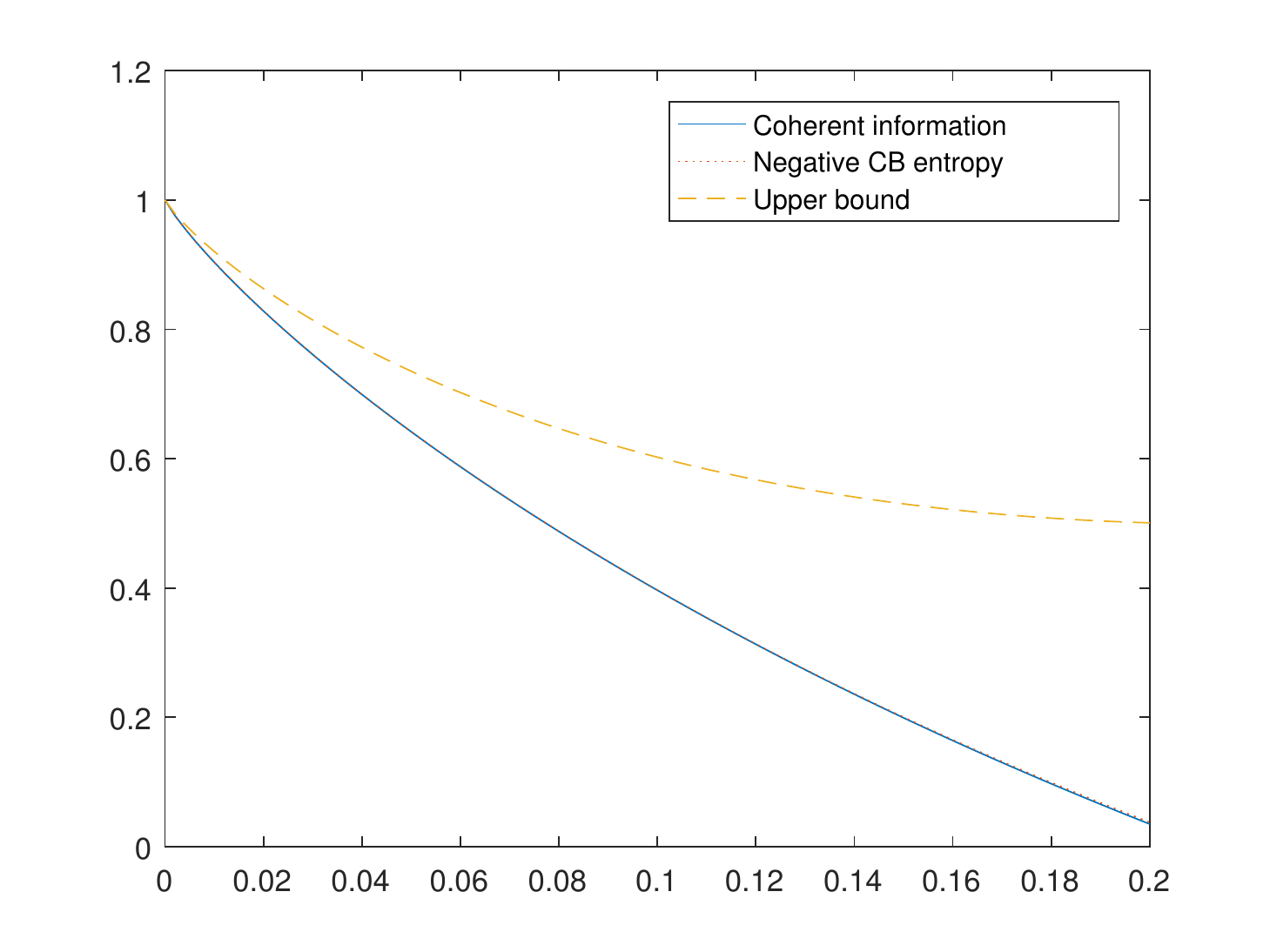}
\end{center}
\caption{Upper and lower bounds on the secret-key-agreement capacity
$P^{\leftrightarrow}(\mathcal{N}_{p})$ and the PPT-assisted quantum capacity
$Q^{\mathbf{PPT},\leftrightarrow}(\mathcal{N}_{p})$\ of the channel defined in
(\ref{eqn:channel}). The vertical axis represents rate in either private bits
or qubits per channel use. The horizontal axis represents the value of the
channel parameter $p$.}%
\label{fig:plots}%
\end{figure}

The upper bound in \eqref{eq:qubit-chan-upper-bound} is plotted in
Figure~\ref{fig:plots}. To calculate $E_{\mathbf{PPT}}(A;B)_{\Phi
(\overline{\mathcal{N}}_{p})}$, we have made use of the relative entropy
optimization techniques put forward recently in \cite{FF17}.

We also consider lower bounds on the assisted capacities. Note that both
$P^{\leftrightarrow}(\mathcal{N}_{p})$ and $Q^{\mathbf{PPT},\leftrightarrow
}(\mathcal{N}_{p})$ are bounded from below by the coherent information
\cite{PhysRevA.54.2629}\ and the negative CB-entropy \cite{DJKR06}\ of the
channel (note that the latter is sometimes called \textquotedblleft reverse
coherent information\textquotedblright). These lower bounds are a direct
consequence of the developments in \cite{DW05}. In Figure~\ref{fig:plots}, we
also plot the coherent information of the channel, with the input state being
the maximally entangled state. For the negative CB-entropy, we optimize over
the input states of the channel and can exploit symmetry to simplify this optimization.

Here we elaborate on how to simplify the calculation of the negative
CB-entropy for our example. As discussed in \cite{DJKR06}, it is possible to
write the negative CB-entropy of a channel $\mathcal{N}_{A\rightarrow B}$ as
the following optimization:%
\begin{equation}
-H_{\operatorname{CB}}(\mathcal{N})=\sup_{\rho}H(B|E)_{\mathcal{U}%
^{\mathcal{N}}(\rho)},
\end{equation}
where $\mathcal{U}_{A\rightarrow BE}^{\mathcal{N}}$ denotes an isometric
channel that extends $\mathcal{N}_{A\rightarrow B}$. Due to the concavity of
conditional entropy \cite{LR73,PhysRevLett.30.434}, it immediately follows
that $H(B|E)_{\mathcal{U}^{\mathcal{N}}(\rho)}$ is concave with respect to the
input density operator $\rho$, and thus the calculation of
$-H_{\operatorname{CB}}(\mathcal{N})$ is a concave optimization problem. For
our example, we can further simplify the calculation of $-H_{\operatorname{CB}%
}(\mathcal{N}_{p})$ by exploiting the symmetry of $\mathcal{N}_{p}$. To see
this symmetry, consider that both the amplitude damping channel and the
depolarizing channel are covariant with respect to $I$ and $Z$, in the sense
that%
\begin{equation}
\mathcal{A}_{p}(U\rho U^{\dag})=U\mathcal{A}_{p}(\rho)U^{\dag},\qquad
\mathcal{D}_{p}(U\rho U^{\dag})=U\mathcal{D}_{p}(\rho)U^{\dag},
\end{equation}
where $U$ can be $I$ or $Z$. This observation then implies that $\mathcal{N}%
_{p}$ is covariant with respect to $I$ and $Z$. By invoking an observation
stated in \cite{Hol06}, it follows that%
\begin{equation}
\mathcal{U}^{\mathcal{N}_{p}}(Z\rho Z)=(Z\otimes\bar{Z})\mathcal{U}%
^{\mathcal{N}_{p}}(\rho)(Z\otimes\bar{Z}),
\end{equation}
where $\mathcal{U}^{\mathcal{N}_{p}}$ is an isometric channel that extends
$\mathcal{N}_{p}$ and $\bar{Z}$ is a unitary representation of $Z$. We can
then exploit the invariance of conditional entropy with respect to local
unitaries and its concavity to find that%
\begin{align}
H(B|E)_{\mathcal{U}^{\mathcal{N}_{p}}(\rho)}  &  =\frac{1}{2}\left[
H(B|E)_{\mathcal{U}^{\mathcal{N}_{p}}(\rho)}+H(B|E)_{(Z\otimes\bar
{Z})\mathcal{U}^{\mathcal{N}_{p}}(\rho)(Z\otimes\bar{Z})}\right] \\
&  =\frac{1}{2}\left[  H(B|E)_{\mathcal{U}^{\mathcal{N}_{p}}(\rho
)}+H(B|E)_{\mathcal{U}^{\mathcal{N}_{p}}(Z\rho Z)}\right] \\
&  \leq H(B|E)_{\mathcal{U}^{\mathcal{N}_{p}}(\frac{1}{2}(\rho+Z\rho Z))}.
\end{align}
Since $\frac{1}{2}(\rho+Z\rho Z)$ has no off-diagonal elements with respect to
the standard basis, the above calculation reduces the optimization of the
negative CB-entropy for $\mathcal{N}_{p}$\ to an optimization over a single parameter.

The optimized negative CB-entropy is plotted in Figure \ref{fig:plots}. Our
findings are consistent with those in earlier works \cite{SSWR14,LLS17}: the
upper bound from \eqref{eq:qubit-chan-upper-bound} is closer to the lower
bounds in the low-noise regime (small values of $p$) than it is in the
high-noise regime.

\section{Generalizations to other resource theories}

In this section, we discuss how to generalize several of the concepts in our
paper to general resource theories \cite{BG15,fritz_2015,RKR15,KR16}. This
generalization has already been considered in the context of the resource
theory of coherence \cite{BGMW17}, and in fact, we note here that the recent
developments in \cite{BGMW17} were what served as the inspiration for our
present paper. In short, a resource theory consists of a few basic
elements.\ There is a set $F$\ of free quantum states, i.e., those that the
players involved are allowed to access without any cost. Related to these,
there is a set of free channels, and they should have the property that a free
state remains free after a free channel acts on it. Once these are defined, it
follows that any state that is not free is considered resourceful, i.e.,
useful in the context of the resource theory. We can also then define a
measure $V$ of the resourcefulness of a quantum state, and some fundamental
properties that it should satisfy are that

\begin{enumerate}
\item it should be monotone non-increasing under the action of a free channel and

\item it should be equal to zero when evaluated on a free state.
\end{enumerate}

\noindent A typical choice of a resourcefulness measure of a state $\rho$
satisfying these requirements is the relative entropy of
resourcefulness:\ $\inf_{\sigma\in F}D(\rho\Vert\sigma)$.

With these basic aspects established and given a measure $V$ of the
resourcefulness of a quantum state, we might be interested in quantifying how
resourceful a channel $\mathcal{N}$\ is. One way of doing so is to define the
amortized resourcefulness of a quantum channel as follows, generalizing the
amortized entanglement from Definition~\ref{def:amortized-ent}:%
\begin{equation}
V_{A}(\mathcal{N})=\sup_{\rho_{RA}}V((\operatorname{id}_{R} \otimes
\mathcal{N}_{A\to B})(\rho_{RA}))-V(\rho_{RA}).
\end{equation}

\begin{figure}[ptb]
\begin{center}
\includegraphics[
width=6.5638in
]{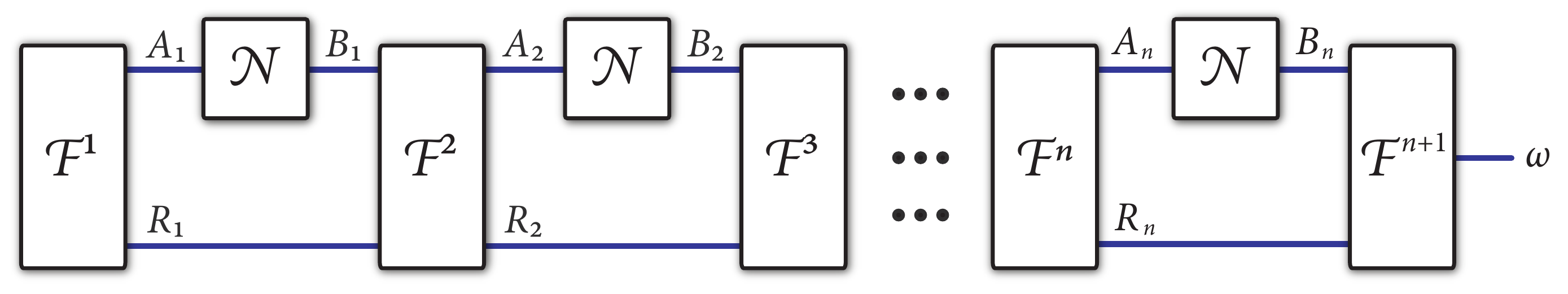}
\end{center}
\caption{A protocol for extracting resourcefulness from a channel
$\mathcal{N}$ by accessing it with free operations between each channel use.}%
\label{fig:resource-theory-protocol}%
\end{figure}

Suppose now that we have a protocol that accesses the channel $\mathcal{N}$ a
total of $n$\ times and between each channel use, we allow for a free channel
to be applied. Such protocols generalize those that we considered in
Sections~\ref{sec:SKA-protocol}\ and \ref{sec:PPT-protocols}. Let $\omega$
denote the final state generated by the protocol, let $\rho_{R_{i}A_{i}}$
denote the state before the $i$th channel use, and let $\sigma_{R_{i}B_{i}}$
denote the state after the $i$th channel use. See
Figure~\ref{fig:resource-theory-protocol} for a depiction of such a protocol.
Then by applying the same reasoning in the proofs of
Propositions~\ref{prop:weak-converse-bound} and
\ref{prop:weak-converse-bound-PPT} (but now using properties 1 and 2 above),
we find the following bound:%
\begin{align}
V(\omega)  &  \leq V(\sigma_{R_{n}B_{n}})\\
&  =V(\sigma_{R_{n}B_{n}})-V(\rho_{R_{1}A_{1}})\\
&  =V((\operatorname{id}_{R}\otimes\mathcal{N}_{A_{n}\rightarrow B_{n}}%
)(\rho_{R_{n}A_{n}}))-V(\rho_{R_{1}A_{1}})+\sum_{i=2}^{n}V(\rho_{R_{i}A_{i}%
})-V(\rho_{R_{i}A_{i}})\\
&  \leq V((\operatorname{id}_{R}\otimes\mathcal{N}_{A_{n}\rightarrow B_{n}%
})(\rho_{R_{n}A_{n}}))-V(\rho_{R_{1}A_{1}})\nonumber\\
&  \qquad+\sum_{i=2}^{n}V((\operatorname{id}_{R}\otimes\mathcal{N}%
_{A_{i-1}\rightarrow B_{i-1}})(\rho_{R_{i-1}A_{i-1}}))-V(\rho_{R_{i}A_{i}})\\
&  =\sum_{i=1}^{n}V((\operatorname{id}_{R}\otimes\mathcal{N}_{A_{i}\rightarrow
B_{i}})(\rho_{R_{i}A_{i}}))-V(\rho_{R_{i}A_{i}})\\
&  \leq nV_{A}(\mathcal{N}), \label{eq:amortized-bound-resource-theory}%
\end{align}
which serves as a limitation on how much of the resource we can extract by
invoking the channel $n$ times in such a way. If $V(\omega)$ can be connected
to meaningful operational parameters such as the closeness of the final state
$\omega$ to a desired target state and the number of basic units of a
resource, as was the case in Propositions~\ref{prop:weak-converse-bound} and
\ref{prop:weak-converse-bound-PPT}, then the above bound would be even more
interesting in the context of a given resource theory.

We can also define $\nu$-freely-simulable channels as a generalization of the
teleportation-simulable channels of \cite{BDSW96,HHH99} and the $\omega
$-PPT-simulable channels introduced in Definition~\ref{def:PPT-sim-chan}:

\begin{definition}
[$\nu$-freely-simulable channel]A quantum channel $\mathcal{N}$ is $\nu
$-freely-simulable if there exists a resourceful state $\nu$ and a free
channel $\mathcal{F}$ such that the following equality holds for all input
states $\rho$:%
\begin{equation}
\mathcal{N}(\rho)=\mathcal{F}(\rho\otimes\nu).
\end{equation}

\end{definition}

\begin{figure}[ptb]
\begin{center}
\includegraphics[
width=6.5638in
]{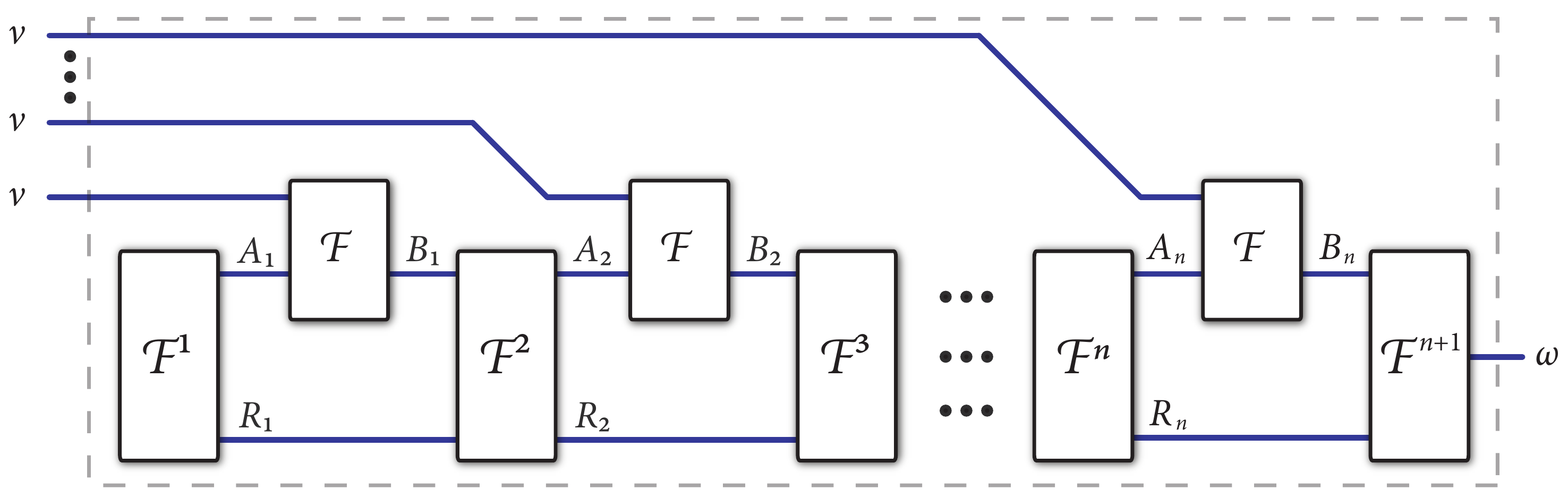}
\end{center}
\caption{A protocol for extracting resourcefulness from a $\nu$-freely
simulable channel $\mathcal{N}$ by accessing it with free operations between
each channel use. Any such protocol can be understood as the action of a
single free channel (everything within the dotted box) acting on the resource
state $\nu^{\otimes n}$.}%
\label{fig:resource-theory-protocol-freely-sim}%
\end{figure}

For $\nu$-freely simulable channels, protocols of the form discussed
previously simplify significantly, as depicted in
Figure~\ref{fig:resource-theory-protocol-freely-sim}. The reduction depicted
in Figure~\ref{fig:resource-theory-protocol-freely-sim} generalizes reduction
by teleportation \cite{BDSW96,Mul12}\ reviewed in
Section~\ref{sec:TP-sim-reduct}\ as well as the more general approach of
\textquotedblleft quantum simulation\textquotedblright\ put forward in
\cite{DM14}, as the reduction applies in the context of any resource theory.
By employing property~1 above and inspecting
Figure~\ref{fig:resource-theory-protocol-freely-sim}, it is immediate that the
following bound holds%
\begin{equation}
V(\omega)\leq V(\nu^{\otimes n}),
\end{equation}
which is just the statement that the amount of resourcefulness that can be
extracted from the channel is limited by the resourcefulness of the underlying
state $\nu$. If the resourcefulness measure is subadditive with respect to
quantum states, then we arrive at the following bound for any protocol of the
above form:%
\begin{equation}
\frac{1}{n}V(\omega)\leq V(\nu).
\end{equation}

We think that it would be very interesting to work out some applications or
consequences of the above observations in the context of several resource
theories, such as thermodynamics \cite{BHORS13}, asymmetry \cite{MS14}, or
non-Gaussianity. We could certainly also consider approximately $\nu
$-freely-simulable channels in order to find bounds on the extraction rates
that are possible from protocols that use resourceful channels that are close
to $\nu$-freely-simulable ones.

\section{Conclusion}

In this paper, we introduced the amortized entanglement of a channel as the
largest difference in entanglement between the output and input of a quantum
channel. We proved several properties of amortized entanglement and considered
special cases of the measures such as amortized relative entropy of
entanglement and amortized Rains relative entropy. One property of especial
interest is the uniform continuity of the latter two special cases, in which
an upper bound on the deviation of the amortized entanglement of two channels
is given in terms of the output dimension of the channels and the diamond norm
of their difference. This uniform continuity bound and the notion of
approximately teleportation- and PPT-simulable channels then immediately leads
to an upper bound on the secret-key-agreement and LOCC-assisted quantum
capacities of such channels. We applied these notions to an example channel,
which consists of a convex mixture of an amplitude damping channel and a
depolarizing channel, and we found that the upper bound is reasonably close to
lower bounds on the capacities whenever the noise in the channel is
sufficiently low. Finally, we discussed how to generalize many of the notions
in the paper to more general resource theories, introducing concepts such as
amortized resourcefulness of a channel and $\nu$-freely-simulable channels.

For future work, we think it would be interesting to explore the
aforementioned generalization further, in the context of other resource
theories such as thermodynamics, asymmetry, or non-Gaussianity.

\bigskip

\textbf{Acknowledgements.} We thank Stefan B\"auml, Siddhartha Das, Nilanjana
Datta, Felix Leditzky, Iman Marvian, and Andreas Winter for discussions
related to this paper. We also acknowledge support from the Office of Naval Research.

\appendix

\section{Supplementary lemmas for uniform continuity}

\label{sec:supp-lemmas}The last inequality in the following lemma was
established in \cite{HHHO05lock} for relative entropy of entanglement, but we
give a different proof in what follows:

\begin{lemma}
\label{lem:dim-bound}Let $E$ refer to either the relative entropy of
entanglement or the Rains relative entropy. For $\rho_{ABC}$ a state, the
following inequality holds%
\begin{align}
E(A;BC)_{\rho}  &  \leq E(A;B)_{\rho}+I(AB;C)_{\rho}\\
&  \leq E(A;B)_{\rho}+2\log_{2}\left\vert C\right\vert .
\end{align}

\end{lemma}

\begin{proof}
Let $\mathcal{S}$ refer to either $\operatorname{SEP}$ or $\operatorname{PPT}%
^{\prime}$. Consider that%
\begin{align}
E(A;BC)_{\rho}  &  =\min_{\sigma_{A:BC}\in\mathcal{S}(A:BC)}D(\rho_{ABC}%
\Vert\sigma_{A:BC})\\
&  \leq\min_{\sigma_{A:B}\in\mathcal{S}(A:B)}D(\rho_{ABC}\Vert\sigma
_{A:B}\otimes\rho_{C})\\
&  =\min_{\sigma_{A:B}\in\mathcal{S}(A:B)}\left[  -H(ABC)_{\rho}%
-\operatorname{Tr}\{\rho_{ABC}\log_{2}(\sigma_{A:B}\otimes\rho_{C})\}\right]
\\
&  =\min_{\sigma_{A:B}\in\mathcal{S}(A:B)}\left[  -H(ABC)_{\rho}%
-\operatorname{Tr}\{\rho_{ABC}\log_{2}(\sigma_{A:B})\}-\operatorname{Tr}%
\{\rho_{ABC}\log_{2}(\rho_{C})\}\right] \\
&  =\min_{\sigma_{A:B}\in\mathcal{S}(A:B)}\left[  -H(ABC)_{\rho}%
-\operatorname{Tr}\{\rho_{AB}\log_{2}(\sigma_{A:B})\}+H(C)_{\rho}\right] \\
&  =\min_{\sigma_{A:B}\in\mathcal{S}(A:B)}\left[  -H(C|AB)_{\rho}-H(AB)_{\rho
}-\operatorname{Tr}\{\rho_{AB}\log_{2}(\sigma_{A:B})\}+H(C)_{\rho}\right] \\
&  =-H(C|AB)_{\rho}+\min_{\sigma_{A:B}\in\mathcal{S}(A:B)}\left[
-H(AB)_{\rho}-\operatorname{Tr}\{\rho_{AB}\log_{2}(\sigma_{A:B})\}\right]
+H(C)_{\rho}\\
&  =-H(C|AB)_{\rho}+E(A;B)_{\rho}+H(C)_{\rho}\\
&  =E(A;B)_{\rho}+I(AB;C)_{\rho}.
\end{align}
This concludes the proof.
\end{proof}

\begin{lemma}
\label{eq:cq-rel-ent}Let $E$ refer to either the relative entropy of
entanglement or the Rains relative entropy. For a classical--quantum state%
\begin{equation}
\rho_{XAB}\equiv\sum_{x}p_{X}(x)|x\rangle\langle x|_{X}\otimes\rho_{AB}^{x},
\end{equation}
the following equality holds%
\begin{equation}
E(A;BX)_{\rho}=\sum_{x}p_{X}(x)E(A;B)_{\rho^{x}}.
\end{equation}

\end{lemma}

\begin{proof}
Let $\mathcal{S}$ refer to either $\operatorname{SEP}$ or $\operatorname{PPT}%
^{\prime}$. Let $\sigma_{A:B}^{x}$ be the positive semi-definite operator that
achieves the minimum for $\rho_{AB}^{x}$ in $E(A;B)_{\rho^{x}}$ and define%
\begin{equation}
\theta_{XAB}\equiv\sum_{x}p_{X}(x)|x\rangle\langle x|_{X}\otimes\sigma
_{A:B}^{x}.
\end{equation}
Then consider that%
\begin{align}
E(A;BX)_{\rho}  &  =\min_{\sigma_{A:BX}\in\mathcal{S}(A:BX)}D(\rho_{ABX}%
\Vert\sigma_{A:BX})\\
&  \leq D(\rho_{ABX}\Vert\theta_{A:BX})\\
&  =\sum_{x}p_{X}(x)D(\rho_{AB}^{x}\Vert\sigma_{A:B}^{x})\\
&  =\sum_{x}p_{X}(x)E(A;B)_{\rho^{x}}.
\end{align}
To see the other inequality, let $\overline{\Delta}_{X}$ be a completely
dephasing channel on system $X$ and consider for any positive semi-definite
operator $\xi_{A:BX}$ in $\mathcal{S}$\ that%
\begin{align}
D(\rho_{ABX}\Vert\xi_{A:BX})  &  \geq D(\overline{\Delta}_{X}(\rho_{ABX}%
)\Vert\overline{\Delta}_{X}(\xi_{A:BX}))\\
&  =D(\rho_{ABX}\Vert\xi_{A:BX}^{\prime})\\
&  =\sum_{x}p_{X}(x)D(\rho_{AB}^{x}\Vert\xi_{A:B}^{x})+D(p_{X}\Vert q_{X})\\
&  \geq\sum_{x}p_{X}(x)E(A;B)_{\rho^{x}},
\end{align}
where%
\begin{equation}
\xi_{A:BX}^{\prime}=\overline{\Delta}_{X}(\xi_{A:BX})=\sum_{x}q_{X}%
(x)|x\rangle\langle x|_{X}\otimes\xi_{A:B}^{x}.
\end{equation}
In the case that $\mathcal{S=}\operatorname{PPT}^{\prime}$, the positive
semi-definite operators are subnormalized \cite{TWW14}, which implies that
$q_{X}(x)$ is a subnormalized probability distribution, and which in turn
implies that $D(p_{X}\Vert q_{X})\geq0$ \cite{W15book}.
\end{proof}

The following lemma, of which the first inequality was proved in
\cite{Linden2005} for relative entropy of entanglement, can be understood as a
direct consequence of Lemmas~\ref{lem:dim-bound}\ and \ref{eq:cq-rel-ent}\ and
the fact that $I(X;AB)_{\rho}\leq H(X)_{\rho}$ for a classical system$~X$:

\begin{lemma}
\label{lem:rel-ent-mixture-ent-bound}Let $E$ refer to either the relative
entropy of entanglement or the Rains relative entropy. Let $\{p_{X}%
(x),\rho_{AB}^{x}\}$ be an ensemble of states and let $\overline{\rho}%
_{AB}\equiv\sum_{x}p_{X}(x)\rho_{AB}^{x}$.\ Then we have that%
\begin{align}
\sum_{x}p_{X}(x)E(A;B)_{\rho^{x}}  &  \leq E(A;B)_{\overline{\rho}%
}+I(X;AB)_{\rho}\\
&  \leq E(A;B)_{\overline{\rho}}+H(X),
\end{align}
where $H(X)$ is the Shannon entropy of the distribution $p_{X}$.
\end{lemma}

\section{Approximately teleportation-simulable channels, approximate
covariance, and channel twirling via teleportation}

\label{sec:cov-TP-sim}

In this appendix, we discuss how the approximately
covariant channels from \cite{LKDW17}\ are approximately teleportation
simulable in the sense of Definition~\ref{def:approx-TP-sim}. We also mention
how the well known protocol of channel twirling \cite{BDSW96} can be
implemented via teleportation over the Choi state of the channel. This latter
result might have applications in other domains, such as randomized
benchmarking \cite{KLRBBJLOSW08}.

We begin with a brief review of some background material.  Let $G$ be a group with unitary representations $g\mapsto U_{A}^{g}$ on
$\mathcal{H}_{A}$ and $g\mapsto V_{B}^{g}$ on $\mathcal{H}_{B}$, respectively.
A quantum channel $\mathcal{N}_{A\rightarrow B}$ is \emph{covariant with
respect to $\{(U_{A}^{g},V_{B}^{g})\}_{g\in G}$} \cite{H02}, if
\[
V_{B}^{g}\,\mathcal{N}(\cdot)V_{B}^{g\dagger}=\mathcal{N}(U_{A}^{g}%
(\cdot)U_{A}^{g\dagger})\qquad\text{for all $g\in G$.}%
\]
A group $G$ is said to form a \emph{unitary one-design}, if there is a unitary
representation $g\mapsto U_{A}^{g}$ of $G$ on $\mathcal{H}_{A}$ such that
\[
\frac{1}{|G|}\sum_{g\in G}U_{A}^{g}\rho_{A}U_{A}^{g\dagger}=\pi_{A}%
\qquad\text{for all states $\rho_{A}$,}%
\]
where $\pi_{A}=\frac{1}{|A|}I_{A}$ denotes the maximally mixed state on
$\mathcal{H}_{A}$.

For a group $G$ with unitary representations $g\mapsto U_{A}^{g}$ on
$\mathcal{H}_{A}$ and $g\mapsto V_{B}^{g}$ on $\mathcal{H}_{B}$, respectively,
and an arbitrary quantum channel $\mathcal{N}_{A\rightarrow B}$, the
\emph{twirled channel} $\mathcal{N}_{G}$ of $\mathcal{N}$ is defined as%
\[
\mathcal{N}_{G}(\cdot)\equiv\frac{1}{|G|}\sum_{g\in G}V_{B}^{g\dagger
}\mathcal{N}(U_{A}^{g}(\cdot)U_{A}^{g\dagger})V_{B}^{g}.
\]
This twirled channel $\mathcal{N}_{G}$ is covariant with respect to
$\{(U_{A}^{g},V_{B}^{g})\}_{g\in G}$ by construction.

The typical way of realizing the twirled channel $\mathcal{N}_{G}$\ is by
means of LOCC. The sender picks $g$ uniformly at random, applies $U_{A}^{g}$
on the input state, sends the state through the channel $\mathcal{N}$,
transmits $g$ to the receiver, who then applies $V_{g}^{\dag}$ at the output.

A different LOCC\ simulation of the twirled channel $\mathcal{N}_{G}$ is
realized by means of a generalized teleportation protocol, as stated in the
following proposition:

\begin{proposition}
The twirled channel $\mathcal{N}_{G}$ can be simulated from $\mathcal{N}%
:\mathcal{L}(\mathcal{H}_{A})\rightarrow\mathcal{L}(\mathcal{H}_{B})$ by means
of a generalized teleportation protocol with a resource state equal to the
Choi state $\omega_{AB}\equiv\mathcal{N}_{A^{\prime\prime}\rightarrow B}%
(\Phi_{AA^{\prime\prime}})$\ where $\mathcal{H}_{A}\simeq\mathcal{H}%
_{A^{\prime}}\simeq\mathcal{H}_{A^{\prime\prime}}$, the POVM\ elements as%
\begin{equation}
\left\{  E_{A^{\prime}A}^{g}\equiv\frac{\left\vert A\right\vert ^{2}%
}{\left\vert G\right\vert }\left(  U_{A^{\prime}}^{g}\right)  ^{\dag}%
\Phi_{A^{\prime}A}U_{A^{\prime}}^{g}\right\}  _{g},
\end{equation}
with $\{U_{A}^{g}\}_{g\in G}$ a one-design, and teleportation correction
operations given by $V_{g}^{\dag}$ acting on the $B$ system. That is, the
following equality holds%
\begin{equation}
\sum_{g}V_{B}^{g\dag}\operatorname{Tr}_{AA^{\prime}}\{E_{A^{\prime}A}^{g}%
(\rho_{A^{\prime}}\otimes\omega_{AB})\}V_{B}^{g}=\frac{1}{\left\vert
G\right\vert }\sum_{g}V_{B}^{g\dag}\mathcal{N}_{A\rightarrow B}(U_{A^{\prime}%
}^{g}\rho_{A}U_{A^{\prime}}^{g\dag})V_{B}^{g}.\label{eq:generalized-TP-twirl}%
\end{equation}

\end{proposition}

\begin{proof}
We follow the proof of \cite[Appendix~A]{WTB16} closely. Let $\mathcal{N}%
:\mathcal{L}(\mathcal{H}_{A})\rightarrow\mathcal{L}(\mathcal{H}_{B})$ be a
quantum channel, and let $G$ be a group with unitary representations
$U_{A}^{g}$ and $V_{B}^{g}$ for $g\in G$, such that%
\begin{equation}
\frac{1}{\left\vert G\right\vert }\sum_{g}U_{A}^{g}X_{A}\left(  U_{A}%
^{g}\right)  ^{\dag}=\operatorname{Tr}\{X_{A}\}\pi_{A}%
,\label{eq:one-design-cond}%
\end{equation}
where $X_{A}\in\mathcal{L}(\mathcal{H}_{A})$ and $\pi$ denotes the maximally
mixed state. (For notational convenience, we are placing the index $g$ as a
superscript.) Consider that%
\begin{equation}
\frac{1}{\left\vert G\right\vert }\sum_{g}U_{A^{\prime}}^{g}\Phi_{A^{\prime}%
A}\left(  U_{A^{\prime}}^{g}\right)  ^{\dag}=\pi_{A^{\prime}}\otimes\pi
_{A},\label{eq:max-ent-cov-action}%
\end{equation}
where $\Phi$ denotes a maximally entangled state and $A^{\prime}$ is a system
isomorphic to $A$. Note that in order for $\{U_{A}^{g}\}$ to satisfy
\eqref{eq:one-design-cond}, it is necessary that $\left\vert A\right\vert
^{2}\leq\left\vert G\right\vert $ \cite{AMTW00}. Consider the POVM
$\{E_{A^{\prime}A}^{g}\}_{g}$, with $A^{\prime}$ a system isomorphic to $A$
and each element $E_{A^{\prime}A}^{g}$ defined as%
\begin{equation}
E_{A^{\prime}A}^{g}\equiv\frac{\left\vert A\right\vert ^{2}}{\left\vert
G\right\vert }\left(  U_{A^{\prime}}^{g}\right)  ^{\dag}\Phi_{A^{\prime}%
A}U_{A^{\prime}}^{g}.
\end{equation}
It follows from the fact that $\left\vert A\right\vert ^{2}\leq\left\vert
G\right\vert $, \eqref{eq:max-ent-cov-action}, and the group property\ that
$\{E_{AA^{\prime}}^{g}\}_{g}$ is a valid POVM.

The simulation of the channel $\mathcal{N}_{G}$ via teleportation begins with
a state $\rho_{A^{\prime}}$ and a shared resource $\omega_{AB}\equiv
\mathcal{N}_{A^{\prime\prime}\rightarrow B}(\Phi_{AA^{\prime\prime}})$. The
desired outcome is for Bob to receive the state $V_{B}^{g\dag}\mathcal{N}%
_{A\rightarrow B}(U_{A^{\prime}}^{g}\rho_{A}U_{A^{\prime}}^{g\dag})V_{B}^{g}$
with probability $1/\left\vert G\right\vert $ and for the protocol to work
independently of the input state $\rho_{A}$. The first step is for Alice to
perform the measurement $\{E_{A^{\prime}A}^{g}\}_{g}$ on systems $A^{\prime}A$
and then send the outcome $g$\ to Bob. Based on the outcome $g$, Bob then
performs $V_{B}^{g\dag}$. The channel realized by such a generalized
teleportation protocol is as follows:%
\begin{equation}
\sum_{g}V_{B}^{g\dag}\operatorname{Tr}_{AA^{\prime}}\{E_{A^{\prime}A}^{g}%
(\rho_{A^{\prime}}\otimes\omega_{AB})\}V_{B}^{g}.\label{eq:generalized-TP}%
\end{equation}

The following analysis demonstrates that this protocol works (i.e., it
simulates $\mathcal{N}_{G}$), by simplifying the form of the post-measurement
state:%
\begin{align}
\left\vert G\right\vert \operatorname{Tr}_{AA^{\prime}}\{E_{A^{\prime}A}%
^{g}(\rho_{A^{\prime}}\otimes\omega_{AB})\} &  =\left\vert A\right\vert
^{2}\operatorname{Tr}_{AA^{\prime}}\{\left(  U_{A^{\prime}}^{g}\right)
^{\dag}|\Phi\rangle_{A^{\prime}A}\langle\Phi|_{A^{\prime}A}U_{A^{\prime}}%
^{g}(\rho_{A^{\prime}}\otimes\omega_{AB})\}\\
&  =\left\vert A\right\vert ^{2}\langle\Phi|_{A^{\prime}A}U_{A^{\prime}}%
^{g}(\rho_{A^{\prime}}\otimes\omega_{AB})\left(  U_{A^{\prime}}^{g}\right)
^{\dag}|\Phi\rangle_{A^{\prime}A}\\
&  =\left\vert A\right\vert ^{2}\langle\Phi|_{A^{\prime}A}U_{A^{\prime}}%
^{g}\rho_{A^{\prime}}\left(  U_{A^{\prime}}^{g}\right)  ^{\dag}\otimes
\mathcal{N}_{A^{\prime\prime}\rightarrow B}(\Phi_{AA^{\prime\prime}}%
))|\Phi\rangle_{A^{\prime}A}\\
&  =\left\vert A\right\vert ^{2}\langle\Phi|_{A^{\prime}A}\left[  U_{A}%
^{g}\rho_{A}\left(  U_{A}^{g}\right)  ^{\dag}\right]  ^{\ast}\mathcal{N}%
_{A^{\prime\prime}\rightarrow B}\left(  \Phi_{AA^{\prime\prime}}\right)
|\Phi\rangle_{A^{\prime}A}.\label{eq:cov-tp-simul-block-1}%
\end{align}
The first three equalities follow by substitution and some rewriting. The
fourth equality follows from the fact that%
\begin{equation}
\langle\Phi|_{A^{\prime}A}M_{A^{\prime}}=\langle\Phi|_{A^{\prime}A}M_{A}%
^{\ast}\label{eq:ricochet-prop}%
\end{equation}
for any operator $M$ and where $\ast$ denotes the complex conjugate, taken
with respect to the basis in which $|\Phi\rangle_{A^{\prime}A}$ is defined.
Continuing, we have that%
\begin{align}
\eqref{eq:cov-tp-simul-block-1} &  =\left\vert A\right\vert \operatorname{Tr}%
_{A}\left\{  \left[  U_{A}^{g}\rho_{A}\left(  U_{A}^{g}\right)  ^{\dag
}\right]  ^{\ast}\mathcal{N}_{A^{\prime\prime}\rightarrow B}\left(
\Phi_{AA^{\prime\prime}}\right)  \right\}  \\
&  =\left\vert A\right\vert \operatorname{Tr}_{A}\left\{  \mathcal{N}%
_{A^{\prime\prime}\rightarrow B}\left(  \left[  U_{A^{\prime\prime}}^{g}%
\rho_{A^{\prime\prime}}\left(  U_{A^{\prime\prime}}^{g}\right)  ^{\dag
}\right]  ^{\dag}\Phi_{AA^{\prime\prime}}\right)  \right\}  \\
&  =\mathcal{N}_{A^{\prime\prime}\rightarrow B}\left(  \left[  U_{A^{\prime
\prime}}^{g}\rho_{A^{\prime\prime}}\left(  U_{A^{\prime\prime}}^{g}\right)
^{\dag}\right]  ^{\dag}\right)  \\
&  =\mathcal{N}_{A^{\prime\prime}\rightarrow B}\left(  U_{A^{\prime\prime}%
}^{g}\rho_{A^{\prime\prime}}\left(  U_{A^{\prime\prime}}^{g}\right)  ^{\dag
}\right)  .
\end{align}
The first equality follows because $\left\vert A\right\vert \langle
\Phi|_{A^{\prime}A}\left(  I_{A^{\prime}}\otimes M_{AB}\right)  |\Phi
\rangle_{A^{\prime}A}=\operatorname{Tr}_{A}\{M_{AB}\}$ for any operator
$M_{AB}$. The second equality follows by applying the conjugate transpose of
\eqref{eq:ricochet-prop}. The above development then implies the following
equality:%
\begin{equation}
\operatorname{Tr}_{AA^{\prime}}\{E_{A^{\prime}A}^{g}(\rho_{A^{\prime}}%
\otimes\omega_{AB})\}=\frac{1}{\left\vert G\right\vert }\mathcal{N}%
_{A^{\prime\prime}\rightarrow B}\left(  U_{A^{\prime\prime}}^{g}%
\rho_{A^{\prime\prime}}\left(  U_{A^{\prime\prime}}^{g}\right)  ^{\dag
}\right)  ,
\end{equation}
which after insertion into \eqref{eq:generalized-TP}, establishes the claim in \eqref{eq:generalized-TP-twirl}.
\end{proof}

The notion of approximate covariance of a quantum channel from \cite{LKDW17} is
based on how close the channel is in diamond norm to its twirled channel:

\begin{definition}
[Approximate covariance \cite{LKDW17}]\label{def:approximate-covariance} Fix a
group $G$ with unitary representations $g\mapsto U_{A}^{g}$ on $\mathcal{H}%
_{A}$ and $g\mapsto V_{B}^{g}$ on $\mathcal{H}_{B}$, respectively. For a given
$\varepsilon\in\lbrack0,1]$, a channel $\mathcal{N}$ $\varepsilon
$\emph{-covariant with respect to $\{(U_{A}^{g},V_{B}^{g})\}_{g\in G}$}, if
\[
\frac{1}{2}\Vert\mathcal{N}-\mathcal{N}_{G}\Vert_{\diamond}\leq\varepsilon.
\]

\end{definition}

It has been known for many years now that a quantum channel covariant with
respect to a one-design is teleportation simulable with resource state equal to the Choi state of the channel \cite[Section~7]{CDP09}.
Thus, an immediate consequence of definitions is that a channel is
$(\varepsilon,\mathcal{N}(\Phi)\mathcal{)}$-approximately
teleportation-simulable if it is $\varepsilon$-covariant, with $\{U_{A}^{g}\}_{g\in
G}$ a one-design.

\bibliographystyle{alpha}
\bibliography{Ref}

\end{document}